\documentclass[runningheads]{llncs}
\usepackage{amssymb}
\usepackage{amsmath}
\usepackage{graphicx}
\usepackage[pdfpagelabels,hypertexnames=false,breaklinks=true,bookmarksopen=true,bookmarksopenlevel=2]{hyperref}

\usepackage{latexsym}
\usepackage{color,graphicx}
\usepackage[fleqn,tbtags]{mathtools}
\usepackage{graphicx}
\usepackage{multirow}
\usepackage{multicol}
\usepackage{stfloats}
\usepackage[centerlast]{caption}
\usepackage{subfloat}
\usepackage{subfigure}
\usepackage{epic}

\newtheorem{clm}{Claim}

\urldef{\mailsa}\path|j.han@qub.ac.uk|
\newcommand{\keywords}[1]{\par\addvspace\baselineskip
\noindent\keywordname\enspace\ignorespaces#1}
\usepackage{a4wide}

\usepackage{tikz}
\usepackage{pgflibraryarrows}
\usepackage{pgflibrarysnakes}

\begin{document}

\mainmatter  

\title{ Oblivious Location-Based Service Query }

\titlerunning{}

%
%
\author{Jinguang Han }
\authorrunning{Han \em{et. al}}

\institute{Centre for Secure Information Technologies (CSIT),\\
Institute of Electronics, Communications and Information Technology (ECIT),\\
Queen’s University Belfast, Belfast, Northern Ireland, BT3 9DT, United
Kingdom\\
\mailsa
}%
%

\toctitle{}
 \tocauthor{} 
  \maketitle


\begin{abstract}
Privacy-preserving location-base services (LBS)  have been proposed to protect users' location privacy. However, there are still some problems in existing schemes: (1) a semi-trusted third party (TTP) is required; or (2) both the computation cost and communication cost to generate a query are linear in the size of the queried area. \\
In this paper, to improve query efficiency, an oblivious location-based service query (OLBSQ) scheme is proposed. Our scheme captures the following features: (1) a semi-trusted TTP is not required; (2) a user can query services from a service provider without revealing her exact location; (3) the service provider can only know the size of a query made by a user; and (4) both the computation cost and the communication cost to generate a query is constant, instead of linear in the size of the queried area. We formalise the definition and security model of OLBSQ schemes. The security of our scheme is reduced to well-known complexity assumptions.
The novelty is to reduce the computation cost and communication cost of making a query and enable the service provider to obliviously and incrementally generate decrypt keys for queried services. This contributes to the growing work of formalising privacy-preserving LBS schemes and improving query efficiency. 

 \keywords{Location-base Services, Location Privacy, Oblivious Transfer, Security}
\end{abstract}

\section{Introduction}\label{sec:intro}

The  advent of mobile devices and mobile networks triggered a new services named  location-based services (LBS). LBS systems enable service providers (SPs) to provide users with accurate services based on their geographical locations. Nowadays, increasing number of users use LBS systems to query nearby Points of Interest (PoI) including shopping centers, restaurants, banks, hospitals, traffic information, navigation, {\em etc.}  However, to query a service, a user must reveal her location to the service provider (SP). Hence, untrusted SPs can profile a user's movement by tracing her location, and conclude her personal information, such as working place, health condition,  commercial partners, {\em etc}. This raises a serious privacy issue.
To protect users' location privacy, privacy-preserving LBS schemes were proposed where either a semi-trusted third party  (TTP) is required or the computation cost of a query is linear in the  size of the queried area. However, in practice, it is difficult to find a party who can work as a semi-trusted TTP in LBS schemes, and mobile devices have constrained computation power and limited storage space.

 Considering the above problems,  an oblivious location-based service query (OLBSQ) scheme is proposed to enhance the security of SPs' services and  protect users' location privacy. Especially, our OLBSQ scheme provides mobile uses with a light query algorithm which has constant computation cost.

\subsection{Related Work}
Due to it can provide accurate services, LBS schemes are becoming increasingly popular. Nevertheless, location privacy has been the primary concern of LBS users. To protect users' location privacy, privacy-preserving LBS schemes were proposed. 

\subsubsection{Privacy-Preserving LBS with A Trusted Third Party}\hfill
\medskip

\noindent In these schemes, to  protect mobile users' location privacy,  a trusted third party called {\em location anonymizer} is required to blur a user's exact location into a cloaked area. Meanwhile, the cloaked area must satisfy the user's privacy requirements. The popular privacy requirement is $k$-anonymity, namely a user's location is indistinguishable from other $k-1$ users' locations.  
Gruteser and  Grunwald  \cite{gg:lbs2003} proposed an anonymous LBS scheme where  the location anonymizer needs to remove any identifiers such as network and address, and perturbs the position data. In \cite{gg:lbs2003}, the location anonymizer knows users' location, and users need to periodically update their location information to the location anonymizer. 

Proposed by Mokbel, Chow and Aref \cite{mca:lbs2006},  $Casper^{*}$ is a privacy-aware query processing method for LBS. In Casper \cite{mca:lbs2006}, the location anonymizer blurs users' exact location into cloaked spatial areas and a privacy-aware query processor is embedded in the database to deal with queries based on the cloaked spatial areas. The privacy-aware query processor supports three types of queries: private queries over public data, public queries over private data and private queries over private data. 

 
Xu and Cai \cite{xc:lbs2007}  addressed the location anonymity issue in  continuous LBS schemes. In \cite{xc:lbs2007}, entropy was used to measure the anonymity degree of a cloaking area, which consider both the number of the users and their anonymity probability distribution in the cloaking area. When issuing a query, a mobile user sends his query and desired anonymity level to the location anonymizer, and then the location anonymizer generates a session identity for the user and contact the service provider to establish a service session. After  a service session is established, the location anonymizer needs to periodically identify a cloaking area for the user according to her latest location, and report the cloaking area to the service provider. Furthermore, a  polynomial time algorithm was proposed to find a cloaking area satisfying the anonymity requirement.

Kalnis {\em et al.} \cite{kgmp:bls2007} proposed a framework to prevent location-based identity inference of users. In \cite{kgmp:bls2007}, when receiving a query, the location anonymizer first removes the user's identity, and uses an anonymizing spatial region to hide the user's location.  This framework optimizes the processing of both location anonymity and spatial queries. 

Gedik and Liu \cite{gl:bls2008} introduced a scalable architecture to protect users' location privacy. The architecture consists of a model of personalised location anonymity and a set of location perturbation algorithms. In \cite{gl:bls2008}, upon receiving a query from a user, the location anonymizer remove the identity of the user and perturbs her location by replacing a 2-dimensional point with a spatial cloaking ranger. Especially, users are allowed to specify the minimum level of anonymity and the maximum temporal and spatial tolerances. 

Chen {\em et al.} \cite{chycdx:bls2018} proposed a new scheme to protect users' location privacy. In \cite{chycdx:bls2018}, redundant point-of-interest (POI) records were applied to protect location privacy. When receiving a query from a user, the location anonymizer first generates a $k$-anonymity  rectangle area for the user, and then sends the anonymous query to the service provider. Notably, a blind filter scheme was proposed to enable the location anonymizer to filter out the redundant POI records on behalf of users. 

 To leveraging spatial diversity in LBS, He {\em et al.} \cite{hjd:lbs2018} first proposed ambient environment-dependent location privacy metrics and a stochastic model, and then developed an optimal stopping-based LBS scheme which enable users to leverage the spatial diversity. 

Grissa  {\em et al.} \cite{gyh:loc2017} proposed two schemes to protect the location privacy of second users where a TTP named fusion centre (FC) is required to orchestrates the sensing operation. The first scheme is based on an order-preserving encryption (OPE) and has lower communication head, while the second scheme is based on a secure comparison protocol and has lesser architectural cost. 

Schlegel {\em et al.} \cite{schw:grid2015} proposed a user-defined privacy LBS scheme called dynamic grid system (DGS) which support both privacy-preserving continuous $k$-nearest-neighbor ($k$-NN) and range queries. In \cite{schw:grid2015}, each user generates a grid structure according to her privacy requirement and embeds it into an encrypted query area. 
When making a query, a user encrypts a secret key $K$ and the grid structure by using an identity-based encryption scheme, and sends the ciphertexts to the service provider. Subsequently, the user generates an encrypted identifier for each cell  in the intended area using a deterministic encryption technique, and sends it to the TTP. To process a query, the service provider decrypts the ciphertext and obtains the secret key and the grid architecture.  The service provider uses the secret key and the deterministic encryption technique to generate encrypted identifiers for all cells where  POIs exist. Later, the service provider sends all the encrypted identifiers to the TTP. The TTP match the encrypted identifiers from the user and those from the service provider, and send the same encrypted identifiers to the user. Finally, the user can decrypt the encrypted identifiers and know the locations of the POIs.
Notably, the communication cost to generate a query is linear with the number of POI in the vicinity  and independent of the number of cells in the grid.

In above schemes, a TTP is required to  protect users' location privacy. However, in practice, it is difficult to find an entity which can play the role of the TTP.

\subsubsection{Privacy-Preserving LBS without A Trusted Third Party}\hfill
\medskip

Chow, Mokbel and Liu \cite{cml:lbs2006} proposed a peer-to-peer (P2P) spatial cloaking scheme which enables users to obtain services without the need of a TTP. Prior to make a query, a user needs to forms a group from her peers via single-hop communication/multiple-hop routing. The spatial cloaked area should cover all peers in the group. Furthermore, the user randomly selects one peer in the group as her agent and sends both her query and cloaked spatial region to the agent. The agent forwards the query to the service provider and receives a list of answers including actual answers and false answers. Then, the agent sends the answers to the user. Finally, the user filter out false answers and obtain the actual answers. The P2P spatial cloaking scheme supports two models: on-demand model and proactive model. Comparatively, the on-demand model is efficient, but requires longer response time. 

Ghinita,  Kalnis and Skiadopoulos \cite{gks:lbs2007} proposed a decentralised LBS scheme named $\mbox{PRIV}\acute{E}$ where each user can organises herself into a hierarchical overlay network and make service queries anonymously. Each user can decide the degree $k$ of anonymity  and the $\mbox{PRIV}\acute{E}$ algorithm can identify an appropriate set  consisting of $k$ users in a distributed manner. To protect users' anonymity, the HILB-ASR algorithm was proposed to guarante that the probability of identifying a real service requester is always bounded by $\frac{1}{k}$.  This scheme is scalable and fault tolerant. 

Paulet {\em et al.} \cite{pkyb:loc2014} proposed a privacy-preserving and content-protecting LBS scheme. This scheme was derived from the oblivious transfer (OT) scheme \cite{np:ot1999} and private information retrieve (PIR) \cite{gr:pir2005}.  Each user firsts runs the OT protocol with the service provider to obtain the location identity and a secret key, and then executes the PIR protocol with the service provider to obtain the location data by using the secret key. The author formalised the security model and analysed the security of the proposed scheme.

Schlegel {\em et al.} \cite{schw:loc2017} proposed an order-retrievable encryption (ORE) scheme with the following two properties: (1) it can generate a encrypted query location; (2) given two encrypted user locations,  a server can determine which one is closed to the an encrypted query location. Subsequently, based on the proposed ORE scheme, a privacy-preserving location sharing services scheme was presented. In \cite{schw:loc2017}, a user or a group  initiator should create a group. The group initiator generates a shared key for the ORE scheme and a shared key for AES scheme. Every user in the group updates periodically her location information to a database server using the ORE and AES  techniques. When receiving a encrypted query location, the server can search out the exact answer without knowing the location information. Finally, the user can use the shared key for AES to decrypt the cipherext and obtain the location information. In \cite{schw:loc2017}, a group of users need to share keys prior to sharing location information.

Hu {\em et al.} \cite{hwhhlc:bls2018} proposed a LBS with query content privacy scheme based on homomorphic encryption, OT and PIR. In \cite{hwhhlc:bls2018}, a user can obtain accurate services, but does not release any  query content information to the server.   The homomorphic encryption is used to compute the Euclidean distance between the attribute vector submitted by a user and the attribute vectors in the database. The OT protocol was used to find the exact match vectors for the queried attribute vector. Finally, the PIR protocol was applied to obtain the intended POI set.  The security  of the proposed scheme was analysed, instead of formal reduction.

In these schemes \cite{cml:lbs2006,gks:lbs2007,pkyb:loc2014,schw:loc2017,hwhhlc:bls2018}, both the computation and communication cost  to generate a query are linear with the size of the queried area. This is undesirable to the devices which have limited computation power and storage space, such as smart phone, tablet, {\em etc}. 

\subsection{Contributions}
To protect users' location privacy, we propose an OLBSQ scheme which can provide the following important features:  (1) a semi-trusted TTP is not required; (2) a user can  query services from a service provider without revealing her exact location; (3) a service provider can only know the size of a query made by a user; and (4) both the computation cost and the communication cost to generate a query is constant, instead of linear with the size of the queried area. 

Our contributions include: (1) both the definition and security model of the proposed OLBSQ scheme are formalised; (2) a concrete OLBSQ scheme is proposed; (3) the security of the proposed OLBSQ is reduced to well-known complexity assumptions.

\subsection{Organization}
The remaining of this paper is organised as follows. Preliminaries used throughout this paper are introduced in Section \ref{sec:preli}. In Section \ref{sec:const}, we formally present our construction. 
In Section \ref{sec:analy}, we prove the security of our scheme. Finally, Section \ref{sec:conc} concludes this paper.

\section{Preliminaries}\label{sec:preli}
In this section, all preliminaries used throughout this paper are introduced. 

\subsection{Formal Definition}

\begin{figure}
\centering
\includegraphics[width=12cm,height=6cm]{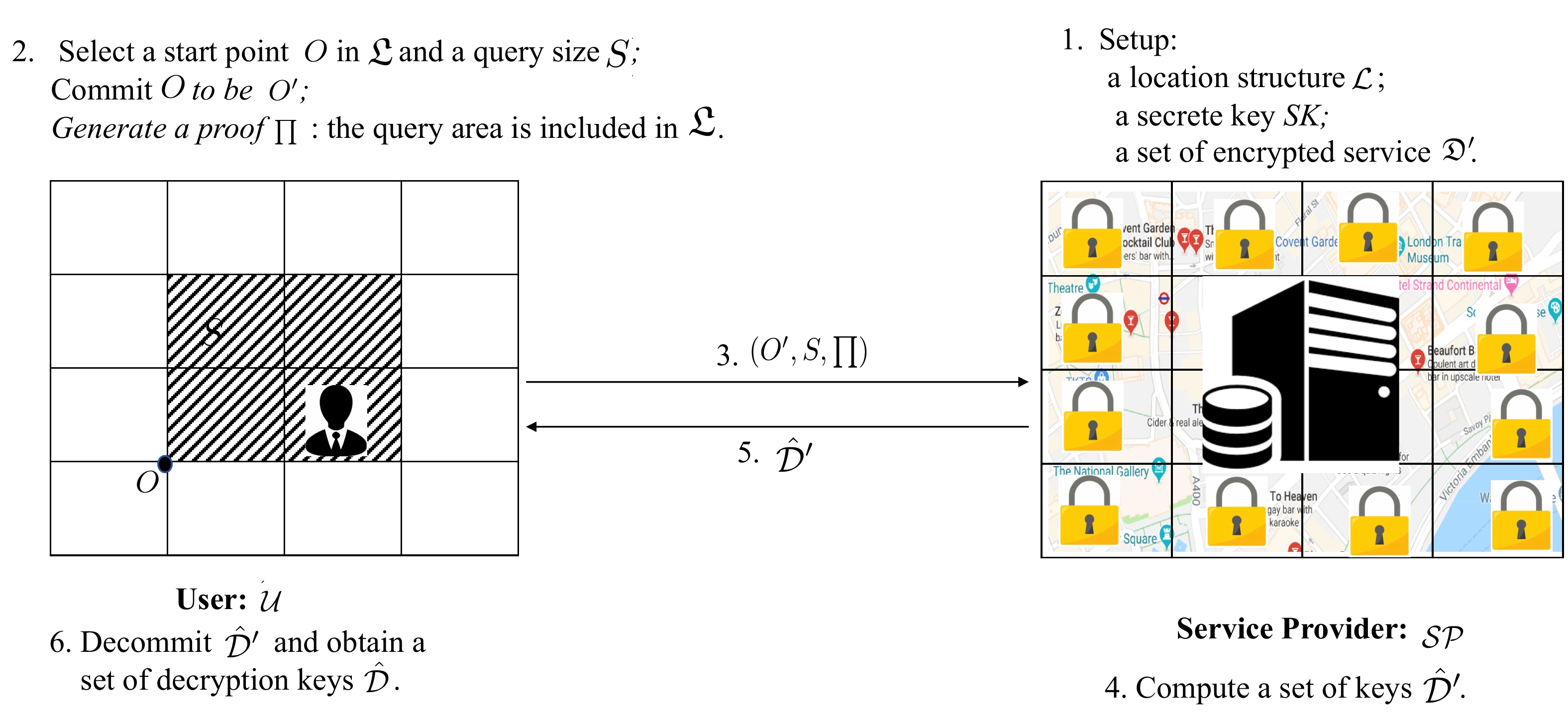}
\caption{The Framework of Our OLBSQ Scheme}\label{framework}
\end{figure}

Let $\mathfrak{L}$ be a location structure (e.g. grid) and $O$ be a point in $\mathfrak{L}$. By $(O;S)$, we denote that the area with start point $O$ and size $S$ in $\mathfrak{L}$. For example, if $\mathfrak{L}$ is a grid system, $(O=(i,j);S=l\times k)$ is the area consisting of the left-bottom point $O$ and $l\times k$ continuous cells.  Let $\mathfrak{D}$ be the services included in $\mathfrak{L}$ and  $\mathfrak{D}'$ be the encrypted services. $\hat{\mathfrak{D}}\in (O;S)$ stands for the services included in the area $(O;S)$. Fig. \ref{framework} describes the framework of our OLBSQ scheme. The service provider $\mathcal{SP}$ first generates a secret key $SK$ and some public parameters $PP$, selects a location structure $\mathcal{L}$. Suppose that $SP$ has a set of service $\mathcal{D}$, he encrypts each service in $\mathcal{D}$ by using $SK$ and its location information, and obtains an encrypted set of services $\mathcal{D}'$. To query services included in an area, a user $\mathcal{U}$ select a start point $O$ and the query size $S$, and then commit $O$ to be a point $O'$. Furthermore, $\mathcal{U}$ generates a proof $\prod$ that the queried area starting from $O$ with size $S$ is included in $\mathcal{L}$. $\mathcal{U}$ sends $(O',S,\prod)$ to $\mathcal{SP}$. If $\prod$ is correct, $\mathcal{SP}$ uses $SK$ to obliviously and incrementally compute a set of keys $\hat{\mathcal{D}}'$ according to $O'$ and $S$, and sends $\hat{\mathcal{D}}'$ to $\mathcal{U}$. Finally, $\mathcal{U}$ decommit $\hat{\mathcal{D}}'$, and obtain a set of decryption key $\hat{\mathcal{D}}$ which enable her to access the intended services. 

\medskip

An OLBSQ scheme consists of the following two algorithms:
\begin{itemize}
\item{\sf Setup}$(1^{\ell},\mathfrak{L},\mathfrak{D})\rightarrow(SK,PP,\mathfrak{D}').$ Taking as input a security parameter $1^{\ell}$, a location structure $\mathfrak{L}$ and a set of services $\mathfrak{D}$, this algorithm outputs a secret key $SK$ for $SP$, some public parameters $PP$ and the encrypted services $\mathfrak{D}'$. 
\item{\sf Service-Transfer}$(\mathcal{U}(O,S,PP)\leftrightarrow \mathcal{SP}(PP,SK))\rightarrow(\hat{\mathfrak{D}},(O',S,\prod))$. This is an interactive algorithm executed between a user $\mathcal{U}$ and the service provider $\mathcal{SP}$. $\mathcal{U}$ takes as input the public parameters $PP$, the start point $O$ and the query size $S$, and outputs the intended services $\hat{\mathfrak{D}}\subset \mathfrak{D}$. $\mathcal{SP}$ takes as input the public parameters $PP$ and the secret key $SK$, and outputs the committed start point $O'$, query size $S$ and a proof $\prod$ that the queried area with start point $O$ and size $S$ is in $\mathcal{L}$. 
\end{itemize}

\begin{definition}
We say that an oblivious location-based service query scheme is correct if and only if
\begin{equation*}
\Pr\left[ \begin{array}{c|l}
& {\sf Setup}(1^{\ell},\mathcal{L},\mathfrak{D})\rightarrow(SK,PP,\mathfrak{D'});\\
\hat{\mathfrak{D}}\subset \mathfrak{D}~\wedge~ \hat{\mathfrak{D}}\in (O,S) &{\sf Service-Transfer}(\mathcal{U}(PP,O,S)\leftrightarrow \\
& \mathcal{SP}(PP,SK))\rightarrow(\hat{\mathfrak{D}},(O',S,\prod));\\
&  \prod ~\mbox{is correct.}
\end{array}
\right]=1.
\end{equation*}
\end{definition}
\subsection{Security Model}

The security model of OLBSQ schemes is formalised by using the simulation-based model \cite{cdn:otac2009,cns:ot2007,j:sim2018,pw:sim2001} where the real world experiment and ideal world experiment are defined. In the real world experiment, there are some parties who run the protocol: an adversary $\mathcal{A}$ who controls some of the parties and an environment $\mathcal{E}$ who provides inputs to all honest parties and interact arbitrarily with $\mathcal{A}$. The dishonest parties are controlled by  $\mathcal{A}$. In the ideal world experiment, there are same parties as in the real world experiment. Notably, these parties do not run the protocol. They submit their inputs to a ideal functionality $\mathcal{F}$ and receive outputs from $\mathcal{F}$. $\mathcal{F}$ specifies the behaviour that the desired protocol should implement in the real world. $\mathcal{E}$ provides inputs to and receives outputs from honest parties. Let $\mathcal{S}$ be a simulator who controls the dishonest parties in the ideal world experiment as $\mathcal{A}$ does in the real world experiment. Furthermore,  $\mathcal{E}$ interacts with $\mathcal{S}$ arbitrarily. 


\begin{definition}
Let ${\bf Real}_{\mathcal{P},\mathcal{E},\mathcal{A}}$ be the probability with which $\mathcal{E}$ runs the protocol $\mathcal{P}$ with $\mathcal{A}$ and outputs 1 in the real world experiment. Let ${\bf Ideal}_{\mathcal{F},\mathcal{E},\mathcal{S}}$ be the probability with which $\mathcal{E}$ interacts with $\mathcal{S}$ and $\mathcal{F}$, and outputs 1 in the ideal world experiment. We say that the protocol $\mathcal{P}$ securely realizes the functionality $\mathcal{F}$ if 
\begin{equation*}
\left| {\bf Real}_{\mathcal{P},\mathcal{E},\mathcal{A}}-{\bf Ideal}_{\mathcal{F},\mathcal{E},\mathcal{S}}\right|\leq \epsilon(\ell).
\end{equation*}
\end{definition}
\medskip

The ideal functionality of  OLBSQ schemes is formalized in Fig. \ref{fig:fun}.

\begin{figure}
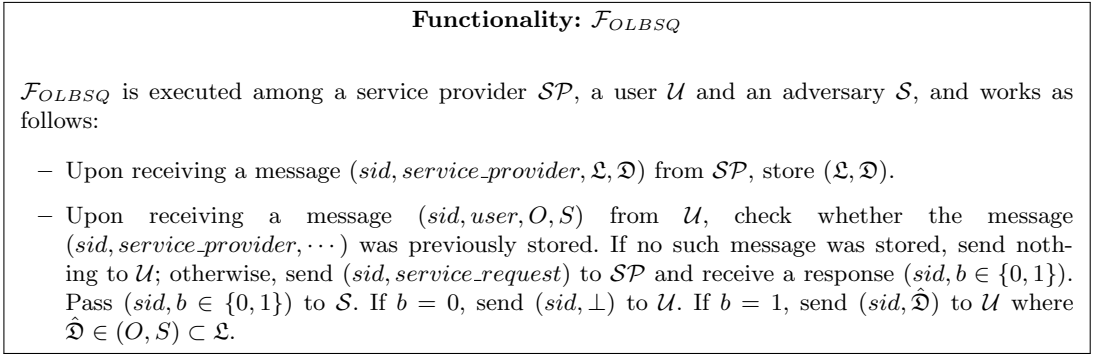

\centering
\fbox{
\begin{minipage}{14cm}
\begin{center} {\bf Functionality: $\mathcal{F}_{OLBSQ}$}\end{center}
\medskip
$\mathcal{F}_{OLBSQ}$ is executed among a service provider $\mathcal{SP}$, a user $\mathcal{U}$ and an adversary $\mathcal{S}$, and works as follows:
\begin{itemize}
\item Upon receiving a message $(sid, service\_provider,\mathfrak{L},\mathfrak{D})$ from $\mathcal{SP}$, store $(\mathfrak{L},\mathfrak{D})$.
\medskip

\item Upon receiving a message $(sid, user, O,S)$ from $\mathcal{U}$, check whether the message $(sid, service\_provider,\cdots)$ was previously stored. If no such message was stored, send nothing to $\mathcal{U}$; otherwise, send $(sid,service\_request)$ to $\mathcal{SP}$ and receive a response $(sid,b\in\{0,1\})$. Pass $(sid,b\in\{0,1\})$ to $\mathcal{S}$. If $b=0$, send $(sid,\perp)$ to $\mathcal{U}$. If $b=1$, send $(sid,\hat{\mathfrak{D}})$ to $\mathcal{U}$ where $\hat{\mathfrak{D}}\in (O,S)\subset \mathfrak{L}$.
\end{itemize}
\end{minipage}
}
\caption{The Functionality of Oblivious Location-Based Service Query Schemes}\label{fig:fun}
\end{figure}

\subsection{Bilinear Map and Complexity Assumptions}
Let $\mathbb{G}_{1}$, $\mathbb{G}_{2}$ and $\mathbb{G}_{\tau}$ be three cyclic groups with prime order $p$.  A map $e:\mathbb{G}_{1}\times\mathbb{G}_{2}\rightarrow\mathbb{G}_{\tau}$ is a bilinear map if it satisfies the following properties:
\begin{enumerate}
\item{\sf Bilinearity.} For all $g\in\mathbb{G}_{1}$, $h\in\mathbb{G}_{2}$ and $x,y\in\mathbb{Z}_{p}$, $e(g^{x},h^{y})=e(g^{y},h^{x})=e(g,h)^{xy}$;
\medskip

\item{\sf Non-degeneracy.} $e(g_{1},g_{2})\neq 1_{\tau}$, where $1_{\tau}$ is the identity of $\mathbb{G}_{\tau}$;
\medskip

\item{Efficiency.} For all  $g\in\mathbb{G}_{1}$ and $h\in\mathbb{G}_{2}$, there is an efficient algorithm to compute $e(g,h)$.
\end{enumerate}
If $\mathbb{G}_{1}=\mathbb{G}_{2}$, $e$ is called a symmetric bilinear map.
Let $\mathcal{BG}(1^{\ell})\rightarrow(e,p,\mathbb{G},\mathbb{G}_{\tau})$ be a generator of symmetric bilinear group which takes as input a security parameter $1^{\ell}$ and outputs a bilinear group $(e,p,\mathbb{G},\mathbb{G}_{\tau})$ with prime order $p$ and $e:\mathbb{G}\times\mathbb{G}\rightarrow\mathbb{G}_{\tau}$.

\begin{definition}{\sf ($q$-Strong Diffie-Hellman ($q$-SDH) Assumption \cite{bb:ss2007}).} Let $\mathcal{BG}(1^{\ell})\rightarrow(e,p,\mathbb{G},\mathbb{G}_{\tau})$ and $\zeta\stackrel{R}{\leftarrow}\mathbb{Z}_{p}$. Suppose that $g$  be a generator of $\mathbb{G}$.  Given $(g,g^{\zeta},g^{\zeta^{2}},\cdots,g^{\zeta^{q}})$, we say that the $q$-SDH assumption holds on the bilinear group $(e,p,\mathbb{G},\mathbb{G}_{\tau})$ if all probable polynomial-time adversarties $\mathcal{A}$ can output $(c,g^{\frac{1}{\zeta+c}})$ with a negligible advantage, namely
\begin{equation*}
Adv_{\mathcal{A}}^{\mbox{q-SDH}}=\left|\Pr[\mathcal{A}(g,g^{\zeta},g^{\zeta^{2}},\cdots,g^{\zeta^{q}})\rightarrow (c,g^{\frac{1}{\zeta+c}})\right|\leq \epsilon(\ell)
\end{equation*}
where $c\stackrel{R}{\leftarrow}\mathbb{Z}_{p}$ and $c\neq -\zeta$.
\end{definition}

\begin{definition}{\sf ( $q$-Power Decisional Diffie-Hellman  ($q$-PDDH) Assumption \cite{cns:ot2007}).} Let $\mathcal{BG}(1^{\ell})\rightarrow(e,p,$ $\mathbb{G},\mathbb{G}_{\tau})$, $g$ be a generator of $\mathbb{G}$ and $\zeta\stackrel{R}{\leftarrow}\mathbb{Z}_{p}$. Given $(g,g^{\zeta},g^{\zeta^{2}},\cdots,g^{\zeta^{q}},H)$, we say that  $q$-PDDH assumption holds on $(e,p,\mathbb{G},\mathbb{G}_{\tau})$ if all probable polynomial-time adversary $\mathcal{A}$ can distinguish $T=(H^{\zeta},H^{\zeta^{2}},\cdots,H^{\zeta^{q}})$ from $T=(\tilde{H}_{1},\tilde{H}_{2},\cdots,\tilde{H}_{q})$ with a negligible advantage, namely 
\begin{equation*}
\begin{split}
Adv_{\mathcal{A}}^{\mbox{q-PDDH}}=& \Big|\Pr[\mathcal{A}(g,g^{\zeta},g^{\zeta^{2}},\cdots,g^{\zeta^{q}},H,H^{\zeta},H^{\zeta^{2}},\cdots,H^{\zeta^{q}})=1]-\\
& \Pr[\mathcal{A}(g,g^{\zeta},g^{\zeta^{2}},\cdots,g^{\zeta^{q}},H,\tilde{H}_{1},\tilde{H}_{2},\cdots,\tilde{H}_{q})=1]\Big|\leq \epsilon(\ell)
\end{split}
\end{equation*}
where $H,\tilde{H}_{1},\tilde{H}_{2},\cdots,\tilde{H}_{q}\stackrel{R}{\leftarrow}\mathbb{G}_{\tau}$.
\end{definition}

\section{Construction}\label{sec:const}
In this section, we describe the formal construction of our OLBQS scheme. 

\subsection{High-Level Overview}

To construct our scheme, we use the grid structure which is described in Fig. \ref{grid}. The location of each cell is determined by the coordinate of the point at its upper-right corner. Suppose that all services included in a cell are encrypted under a same key. Firstly, the service provider divides the whole area into $m\times n$ cells, and then  generates a secret key and some public parameters. The service provider encrypts each service in a cell by using his secret key and the coordinate of the cell. Finally, the service provider publishes the public parameters and the encrypted services. 

When making a service query, a user selects a start point  $O=(i,j)$ and the query size $S=k\times l$ where $k$ and $l$ are the numbers of cells in each row and  each column, respectively. The user commits $O=(i,j)$ to be a point $O'$, generates a proof $\prod$ that the queried area $(O';S)$ is included in $\mathcal{L}$, and sends $(O', S,\prod_{U})$ to the service provider. After receiving $(O', S,\prod_{U})$, the service provider first checks the correctness of $\prod_{U}$, and then uses his secret key to obliviously an incrementally compute a set of keys according $O'$ and $S$. Furthermore, the service provider generates a proof $\prod_{SP}$ that these  keys are computed correctly, and sends the keys and $\prod_{SP}$ to the user. Finally, the user verifies the proof $\prod_{SP}$, de-commits the  keys and obtains the corresponding decryption keys. Finally, the user decrypts the ciphertexts and obtains the intended services. Notably, to retrieve a service, the user only needs to execute 3 exponent operations on $\mathbb{G}_{\tau}$.

\begin{figure}
\centering
\includegraphics[width=10cm,height=6cm]{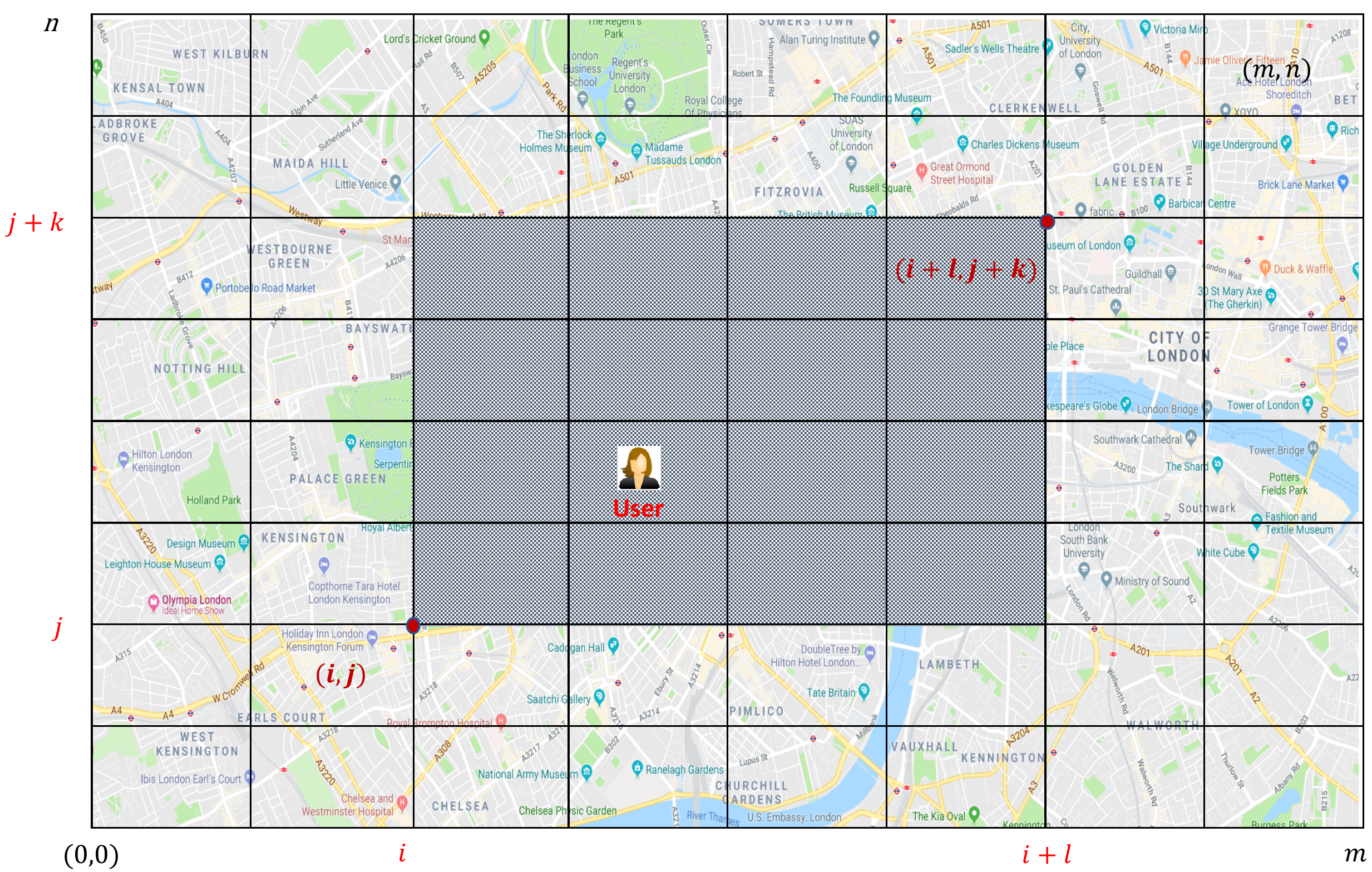}\caption{Grid Location Model of Our Scheme}\label{grid}
\end{figure}

\subsection{Our Construction}
Our OLBSQ scheme is presented in Fig. \ref{fig:setup} and Fig. \ref{fig:service_tr}.
\medskip

\noindent{\bf Setup.} The service provider $\mathcal{SP}$ first divides the whole area $\mathcal{L}$ into $m\times n$ cells. $\mathcal{SP}$ generates a bilinear group by running $\mathcal{BG}(1^{\ell})\rightarrow(e,p,\mathbb{G},\mathbb{G}_{\tau})$, and then selects its secret key $SK=(\alpha_{1},\alpha_{2},\beta_{1},\beta_{2},x,y,\mathfrak{h})$ where $\alpha_{1},\alpha_{2},\beta_{1},\beta_{2},x,y\stackrel{R}{\leftarrow}\mathbb{Z}_{p}$ and $\mathfrak{h}\stackrel{R}{\leftarrow}\mathbb{G}_{2}$. To encrypt the service  $M_{i,j}$ in a cell $C(i,j)$ using its coordinate $(i,j)$, $\mathcal{SP}$ computes $A_{i,j}=g_{1}^{i}h_{1}^{j}g_{2}^{x^{i}}h_{2}^{y_{j}}$ and $B_{i,j}=e(A_{i,j},\mathfrak{h})\cdot M_{i,j}$ for $i=1,2,\cdots,m$ and $j=1,2,\cdots,n$. To enable each user $\mathcal{U}$ to prove that a committed point is in the whole area and $\mathcal{SP}$ to obliviously and incrementally generate decryption keys according $\mathcal{U}$'s query, $\mathcal{SP}$ computes $H=e(\mathfrak{g},\mathfrak{h})$,
 $W_{1}=g_{1}^{\alpha_{1}}$, $W_{2}=g_{2}^{\alpha_{2}}$, $W_{1}'=h_{1}^{\beta_{1}}$, $W_{2}'=h_{2}^{\beta_{2}}$,  $ \Gamma_{1}^{i}=g_{1}^{\frac{1}{\alpha_{1}+i}}$, $\Gamma_{2}^{j}=h_{1}^{\frac{1}{\beta_{1}+j}}$, $(C_{i,1}=g_{2}^{x^{i}},C_{i,2}=g_{2}^{\frac{1}{\alpha_{2}+x^{i}}},C_{i,3}=e(\mathfrak{g},\mathfrak{h})^{x^{i}})$, $(D_{j,1}=h_{2}^{y^{j}},D_{j,2}=h_{2}^{\frac{1}{\beta_{2}+y^{j}}},D_{j,3}=e(\mathfrak{g},\mathfrak{h})^{y^{j}})$ for $i=0,1,2,\cdots,m$ and $j=1,2,\cdots,n$. Actually, $(W_{1},W_{2},W'_{1},W'_{2},\Gamma_{1}^{i},\Gamma_{2}^{j},C_{i,2},D_{j,2})$ are used by $\mathcal{U}$ to prove that a committed start point $O(i,j)$ is within $\mathcal{L}$ for $i=1,2,\cdots,m$ and $j=1,2,\cdots,n$; while other parameters are used by $\mathcal{SP}$ to computes decryption keys. Finally, the public parameters are $PP=\Big(e,p,\mathbb{G},\mathbb{G}_{\tau},\mathfrak{g},g_{1},g_{2},g_{3},g_{4},H,W_{1},$ $W_{2},W'_{1},W'_{2},\Gamma_{1}^{1},\cdots,\Gamma_{1}^{m},\Gamma_{2}^{1},\cdots,\Gamma_{2}^{n},((A_{1,1},B_{1,1}),\cdots,\\(A_{m,n},B_{m,n}),(C_{1,1},$ $C_{1,2},C_{1,3}),\cdots,(C_{m,1},$ $C_{m,2},C_{m,3}),(D_{1,1},D_{2,1},D_{1,3}),\cdots,(D_{n,1},D_{n,2},D_{n,3})\Big)$ and $\mathcal{D}'=\left\{((A_{i,j}, B_{i,j})_{i=1}^{m})_{j=1}^{n}\right\}$.

\begin{figure}
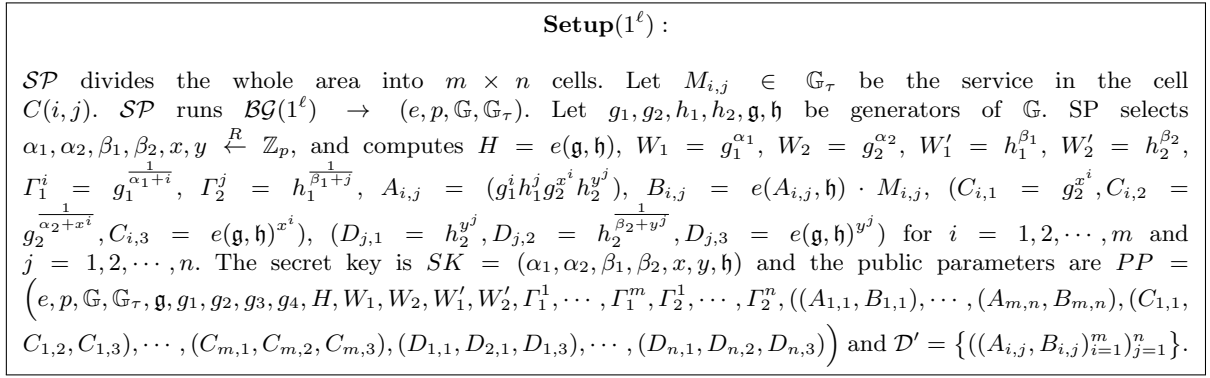

\centering
\fbox{
\begin{minipage}{15.5cm}
\begin{center}{\bf Setup}$(1^{\ell}):$\end{center}
$\mathcal{SP}$ divides the whole area into $m\times n$ cells. Let   $M_{i,j}\in\mathbb{G}_{\tau}$ be the service  in the  cell $C(i,j)$. $\mathcal{SP}$  runs $\mathcal{BG}(1^{\ell})\rightarrow(e,p,\mathbb{G},\mathbb{G}_{\tau})$. Let $g_{1},g_{2},h_{1},h_{2},\mathfrak{g},\mathfrak{h}$ be  generators of $\mathbb{G}$. SP selects $\alpha_{1},\alpha_{2},\beta_{1},\beta_{2},x,y\stackrel{R}{\leftarrow}\mathbb{Z}_{p}$, and computes $H=e(\mathfrak{g},\mathfrak{h})$,
 $W_{1}=g_{1}^{\alpha_{1}}$, $W_{2}=g_{2}^{\alpha_{2}}$, $W_{1}'=h_{1}^{\beta_{1}}$, $W_{2}'=h_{2}^{\beta_{2}}$,  $ \Gamma_{1}^{i}=g_{1}^{\frac{1}{\alpha_{1}+i}}$, $\Gamma_{2}^{j}=h_{1}^{\frac{1}{\beta_{1}+j}}$, $A_{i,j}=(g_{1}^{i}h_{1}^{j}g_{2}^{x^{i}}h_{2}^{y^{j}})$, $B_{i,j}=e(A_{i,j},\mathfrak{h})\cdot M_{i,j}$, $(C_{i,1}=g_{2}^{x^{i}},C_{i,2}=g_{2}^{\frac{1}{\alpha_{2}+x^{i}}},C_{i,3}=e(\mathfrak{g},\mathfrak{h})^{x^{i}})$, $(D_{j,1}=h_{2}^{y^{j}},D_{j,2}=h_{2}^{\frac{1}{\beta_{2}+y^{j}}},D_{j,3}=e(\mathfrak{g},\mathfrak{h})^{y^{j}})$ for $i=1,2,\cdots,m$ and $j=1,2,\cdots,n$.  The secret key is $SK=(\alpha_{1},\alpha_{2},\beta_{1},\beta_{2},x,y,\mathfrak{h})$ and the public parameters  are 
$PP=\Big(e,p,\mathbb{G},\mathbb{G}_{\tau},\mathfrak{g},g_{1},g_{2},g_{3},g_{4},H,W_{1},W_{2},W'_{1},W'_{2},\Gamma_{1}^{1},\cdots,\Gamma_{1}^{m},\Gamma_{2}^{1},\cdots,\Gamma_{2}^{n},((A_{1,1},B_{1,1}),\cdots,(A_{m,n},B_{m,n}),(C_{1,1},$ $C_{1,2},C_{1,3}),\cdots,(C_{m,1},C_{m,2},C_{m,3}),(D_{1,1},D_{2,1},D_{1,3}),\cdots,(D_{n,1},D_{n,2},D_{n,3})\Big)$ and $\mathcal{D}'=\left\{((A_{i,j}, B_{i,j})_{i=1}^{m})_{j=1}^{n}\right\}$.
\end{minipage}
}\caption{Setup Algorithm}\label{fig:setup}
\end{figure}
\medskip

\begin{figure}
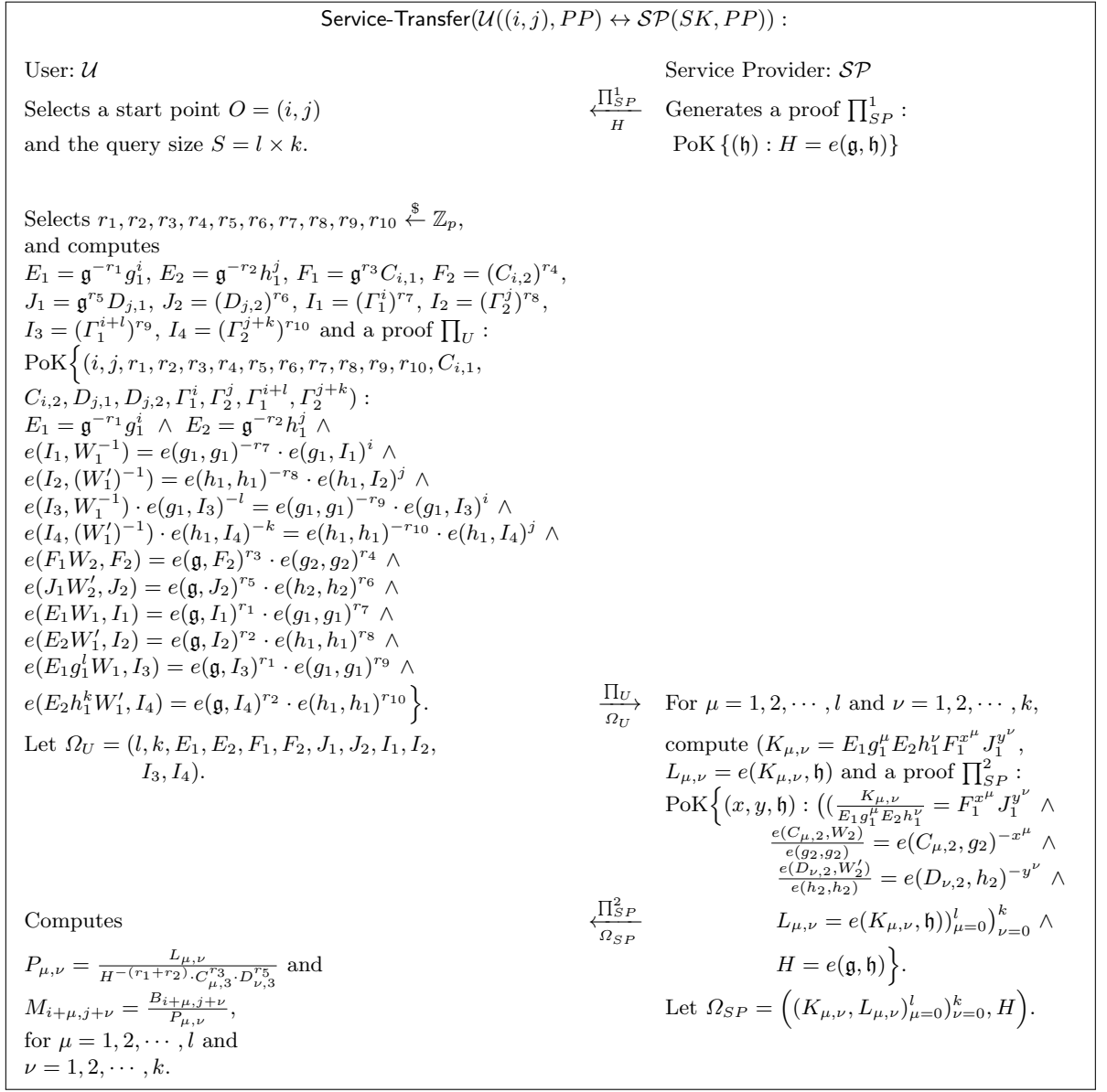

\centering
\fbox{
\begin{minipage}{15.5cm}
\begin{center}{\sf Service-Transfer}$(\mathcal{U}((i,j),PP)\leftrightarrow \mathcal{SP}(SK,PP)):$\end{center}
\begin{tabular}{lcl}
User: $\mathcal{U}$ &  & Service Provider: $\mathcal{SP}$\\
Selects a start point $O=(i,j)$   & ~ $\xleftarrow[H]{\prod_{SP}^{1}}$ ~ & Generates a proof $\prod_{SP}^{1}:$\\
  and the query size $S=l\times k$. & & ~$ \mbox{PoK}\left\{(\mathfrak{h}):H=e(\mathfrak{g},\mathfrak{h})\right\}$\\
   \medskip\\
    
Selects $r_{1},r_{2},r_{3},r_{4},r_{5},r_{6},r_{7},r_{8},r_{9},r_{10}\stackrel{\$}{\leftarrow}\mathbb{Z}_{p}$, \\
 and computes\\
$E_{1}=\mathfrak{g}^{-r_{1}}g_{1}^{i}$, $E_{2}=\mathfrak{g}^{-r_{2}}h_{1}^{j}$,
 $F_{1}=\mathfrak{g}^{r_{3}}C_{i,1}$,  $F_{2}=(C_{i,2})^{r_{4}}$,\\
 $J_{1}=\mathfrak{g}^{r_{5}}D_{j,1}$, $J_{2}=(D_{j,2})^{r_{6}}$,
   $I_{1}=(\Gamma_{1}^{i})^{r_{7}}$, $I_{2}=(\Gamma_{2}^{j})^{r_{8}}$,  &  &\\
   $I_{3}=(\Gamma_{1}^{i+l})^{r_{9}}$, $I_{4}=(\Gamma_{2}^{j+k})^{r_{10}}$ 
and a proof $\prod_{U}:$\\
$ \mbox{PoK}\Big\{(i,j,r_{1},r_{2},r_{3},r_{4},r_{5},r_{6},r_{7},r_{8},r_{9},r_{10},C_{i,1},$\\
$C_{i,2},D_{j,1},D_{j,2},\Gamma_{1}^{i},\Gamma_{2}^{j},\Gamma_{1}^{i+l},\Gamma_{2}^{j+k}):$\\
$E_{1}=\mathfrak{g}^{-r_{1}}g_{1}^{i}~ \wedge~ E_{2}=\mathfrak{g}^{-r_{2}}h_{1}^{j}~ \wedge$\\
$e(I_{1},W_{1}^{-1})=e(g_{1},g_{1})^{-r_{7}}\cdot e(g_{1},I_{1})^{i}~\wedge$\\
$e(I_{2},(W'_{1})^{-1})=e(h_{1},h_{1})^{-r_{8}}\cdot e(h_{1},I_{2})^{j}~\wedge$\\
$e(I_{3},W_{1}^{-1})\cdot e(g_{1},I_{3})^{-l}=e(g_{1},g_{1})^{-r_{9}}\cdot e(g_{1},I_{3})^{i}~\wedge$\\
$e(I_{4},(W'_{1})^{-1})\cdot e(h_{1},I_{4})^{-k}=e(h_{1},h_{1})^{-r_{10}}\cdot e(h_{1},I_{4})^{j}~\wedge$\\
$ e(F_{1}W_{2},F_{2})=e(\mathfrak{g},F_{2})^{r_{3}}\cdot e(g_{2},g_{2})^{r_{4}}~\wedge$\\
 $e(J_{1}W'_{2},J_{2})=e(\mathfrak{g},J_{2})^{r_{5}}\cdot e(h_{2},h_{2})^{r_{6}}~\wedge$\\
$e(E_{1}W_{1},I_{1})=e(\mathfrak{g},I_{1})^{r_{1}}\cdot e(g_{1},g_{1})^{r_{7}}~ \wedge$\\
 $ e(E_{2}W'_{1},I_{2})=$ $ e(\mathfrak{g},I_{2})^{r_{2}}\cdot e(h_{1},h_{1})^{r_{8}} ~\wedge $\\
  $e(E_{1}g_{1}^{l}W_{1},I_{3})=e(\mathfrak{g},I_{3})^{r_{1}}\cdot e(g_{1},g_{1})^{r_{9}} ~\wedge $\\
 $e(E_{2}h_{1}^{k}W'_{1},I_{4})= e(\mathfrak{g},I_{4})^{r_{2}}\cdot e(h_{1},h_{1})^{r_{10}} \Big\}.$ &
~~$\xrightarrow[\Omega_{U}]{\prod_{U}}$~~&For $\mu=1,2,\cdots,l$ and $\nu=1,2,\cdots,k$,   \\
Let $\Omega_{U}=(l,k,E_{1},E_{2},F_{1},F_{2},J_{1},J_{2},I_{1},I_{2},$
&& compute $(K_{\mu,\nu}=E_{1}g_{1}^{\mu}E_{2}h_{1}^{\nu}F_{1}^{x^{\mu}}J_{1}^{y^{\nu}},$  \\
\hspace{1.7cm}$I_{3},I_{4})$. & & $L_{\mu,\nu}=e(K_{\mu,\nu},\mathfrak{h})$  and a proof $\prod_{SP}^{2}:$\\
& & $ \mbox{PoK}\Big\{(x,y,\mathfrak{h}):\big((\frac{K_{\mu,\nu}}{E_{1}g_{1}^{\mu}E_{2}h_{1}^{\nu}}=F_{1}^{x^{\mu}}J_{1}^{y^{\nu}}~\wedge $\\
& & \hspace{1.5cm}$ \frac{e(C_{\mu,2},W_{2})}{e(g_{2},g_{2})}=e(C_{\mu,2},g_{2})^{-x^{\mu}}~\wedge$\\
& & 
\hspace{1.5cm} $\frac{e(D_{\nu,2},W'_{2})}{e(h_{2},h_{2})}=e(D_{\nu,2},h_{2})^{-y^{\nu}}~\wedge$\\
Computes &~ $\xleftarrow[\Omega_{SP}]{\prod_{SP}^{2}}$ ~ & \hspace{1.5cm} $L_{\mu,\nu}=e(K_{\mu,\nu},\mathfrak{h}))_{\mu=0}^{l}\big)_{\nu=0}^{k}~\wedge$\\
$P_{\mu,\nu}=\frac{L_{\mu,\nu}}{H^{-(r_{1}+r_{2})}\cdot C_{\mu,3}^{r_{3}}\cdot D_{\nu,3}^{r_{5}}}$ and 
& & \hspace{1.5cm} $ H=e(\mathfrak{g},\mathfrak{h})\Big\}$.\\
 $M_{i+\mu,j+\nu}=\frac{B_{i+\mu,j+\nu}}{P_{\mu,\nu}}$, 
 & & Let $\Omega_{SP}=\Big((K_{\mu,\nu},L_{\mu,\nu})_{\mu=0}^{l})_{\nu=0}^{k},H\Big)$.\\
 for $\mu=1,2,\cdots,l$ and\\
 $\nu=1,2,\cdots,k$.
\end{tabular}
\end{minipage}}\caption{Service Transfer Algorithm}\label{fig:service_tr}
\end{figure}
\medskip

\noindent{\bf Service-Transfer.} To make a query, $\mathcal{U}$ first selects a start point $O=(i,j)$ and query size $S=l\times k$. $\mathcal{SP}$ generates a proof $\prod_{SP}^{1}$ that he knows the value $\mathfrak{h}$ which is used to encrypt services. If $\prod_{SP}^{1}$ is correct, $\mathcal{U}$ selects  $r_{1},r_{2},r_{3},r_{4},r_{5},r_{6},r_{7},r_{8},r_{9},r_{10}\stackrel{\$}{\leftarrow}\mathbb{Z}_{p}$ and commits  $(i,j,x^{i},y^{j},i+l,j+k)$ into $(E_{1},E_{2},F_{1},F_{2},J_{1},J_{2},I_{1},I_{2},I_{3},I_{4})$. Let $\Omega_{U}=(l,k,E_{1},E_{2},F_{1},F_{2},J_{1},J_{2},I_{1},I_{2},I_{3},I_{4})$. Furthermore, $\mathcal{U}$ generates a proof $\prod_{\mathcal{U}}$ that the query area $(O;S)$ is within $\mathcal{L}$. $\mathcal{U}$ sends $\Omega_{U}$ and $\prod_{U}$ to $\mathcal{SP}$. 

If $\prod_{U}$ is correct, $\mathcal{SP}$ obliviously and incrementally computes a set of keys $(K_{\mu,\nu},L_{\mu,\nu})$ using his secret key $(x,y)$ and generates a proof $\prod_{SP}^{2}$ that $K_{\mu,\nu}$ and $L_{\mu,\nu}$ are generates correctly,  where $\mu=1,2,\cdots,l$ and $\nu=1,2,\cdots,k$.  Let $\Omega_{SP}=\left((K_{\mu,\nu},L_{\mu,\nu})_{\mu=1}^{l})_{\nu=1}^{k},H\right)$. 
$\mathcal{SP}$ sends $\Omega_{SP}$ and $\prod_{SP}^{2}$ to $\mathcal{U}$.

If $\prod_{SP}^{2}$ is correct, $\mathcal{U}$ uses $(r_{1},r_{2},r_{3},r_{5})$ to de-commit the  key $(K_{\mu,\nu},L_{\mu,\nu})$ and obtain $P_{\mu,\nu}=e(g_{1}^{i+\mu}h_{1}^{j+\nu}g_{2}^{x^{i+\mu}}h_{2}^{y^{j+\nu}},\mathfrak{h})$. Furthermore, $\mathcal{U}$ can obtain the services by computing $M_{i+\mu,j+\nu}=\frac{B_{i+\mu,j+\nu}}{P_{\mu,\nu}}$, where $\mu=1,2,\cdots,m$ and $\nu=1,2,\cdots,n$.

\subsection{Efficiency Analysis}
The computation cost and communication cost of our OLBSQ scheme are presented in Table \ref{tab:comp} and Table \ref{tab:comm}, respectively. By $\mathbb{E}$, $\mathbb{E}_{\tau}$, $\mathbb{P}$, $\mathbb{H}$, we denote the time of executing one exponent on the group $\mathbb{G}$, executing one exponent on the group $\mathbb{G}_{\tau}$, executing a pairing and executing one hash function, respectively. $\mathbb{E}_{\mathbb{G}}$, $\mathbb{E}_{\mathbb{G}_{\tau}}$ and $\mathbb{E}_{\mathbb{Z}_{p}}$ stand for the size of one element in the group $\mathbb{G}$, $\mathbb{G}_{\tau}$ and $\mathbb{Z}_{p}$, respectively. 
\begin{table}
\caption{Computation Cost of Our OLBSQ Scheme}
\centering
\begin{tabular}{|c|c|c|c|c|}
\hline
\multirow{3}{*}{Algorithm}& \multirow{3}{*}{Setup} & \multicolumn{3}{|c|}{ Service Transfer}\\
\cline{3-5}
& & \multicolumn{2}{|c|}{$\mathcal{U}$} & \multirow{2}{*}{$\mathcal{SP}$}\\
\cline{3-4}
& & Query& Retrieve &\\
\hline
\multirow{2}{*}{Computation Cost}  & $(4+3m+3n+4mn)\mathbb{E}$ & $16\mathbb{E}+17\mathbb{E}_{\tau}$ & $3klE_{\mathbb{G}}+2(l+k+lk)\mathbb{E}_{\tau} $& $(11+3kl)\mathbb{E}+(33+2kl)\mathbb{E}_{\tau}$ \\
& $+(m+n)\mathbb{E}_{\tau}+(1+mn)\mathbb{P}$ & $+15\mathbb{P}+13\mathbb{H}$ & $2(l+k+lk)\mathbb{P}+2kl\mathbb{H}$& $+(27+4kl)\mathbb{P}+(13+2kl)\mathbb{H}$\\
\hline
\end{tabular}\label{tab:comp}
\end{table}

\begin{table}
\caption{Communication Cost of Our OLBSQ Scheme}
\centering
\begin{tabular}{|c|c|c|c|}
\hline
\multirow{2}{*}{Algorithm}& \multirow{2}{*}{Setup} & \multicolumn{2}{|c|}{Service Transfer}\\
\cline{3-4}
& & $\mathcal{U}\rightarrow\mathcal{SP}$ & $\mathcal{U}\leftarrow\mathcal{SP}$ \\
\hline
\multirow{2}{*}{Communication Cost}  & $(10+3m+3n+mn)E_{\mathbb{G}}+$ & \multirow{2}{*}{$12E_{\mathbb{G}}+16E_{\mathbb{G}_{\tau}}+36E_{\mathbb{Z}_{p}}$} & $(1+3kl)E_{\mathbb{G}}+(2+2kl+l+k)E_{\mathbb{G}_{\tau}}$\\
& $(1+m+n+mn)E_{\mathbb{G}_{\tau}}$ & & $+(1+4kl)E_{\mathbb{Z}_{p}}$\\
\hline
\end{tabular}\label{tab:comm}
\end{table}

\section{Security Analysis}\label{sec:analy}

In this section, the security of our OLBSQ scheme described in Fig. \ref{fig:setup} and Fig. \ref{fig:service_tr} is proven.

\begin{theorem}
 Our oblivious location-based service query scheme in Fig. \ref{fig:setup} and Fig. \ref{fig:service_tr} securely realize the functionality $\mathcal{F}_{OLBSQ}$ in Fig. \ref{fig:fun} under the $q$-SDH and $q$-PDDH assumptions. \label{theo:1}
\end{theorem}

To prove Theorem \ref{theo:1}, we consider the cases where either the user or the service provider is corrupted. We show that there exists a simulator $\mathcal{S}$ such that it can interact with the ideal functionality $\mathcal{F}_{OLBSQ}$ (simply denoted as $\mathcal{F}$) and the environment $\mathcal{E}$ appropriately and ${\bf Real}_{\mathcal{P},\mathcal{E},\mathcal{A}}$ and ${\bf Ideal}_{\mathcal{F},\mathcal{E},\mathcal{S}}$ are indistinguishable. 

In order to prove the indistinguishability between ${\bf Real}_{\mathcal{P},\mathcal{E},\mathcal{A}}$ and ${\bf Ideal}_{\mathcal{F},\mathcal{E},\mathcal{S}}$, a sequence of hybrid games {\bf Game}$_{0}$, {\bf Game}$_{1}$, $\cdots$, {\bf Game}$_{n'}$ are defined. For each {\bf Game}$_{i}$, we show that there exists a simulator $Sim_{i}$ that runs $\mathcal{A}$ as a subroutine and provides $\mathcal{E}$'s view, for $i=1,2,\cdots,n'$. {\bf Hybrid}$_{\mathcal{E},Sim_{i}}(\ell)$ stands for the probability that $\mathcal{E}$ outputs $1$ running in the world provided by $Sim_{i}$. $Sim_{0}$ runs $\mathcal{A}$ and other honest parties in the real-world experiment, so {\bf Hybrid}$_{\mathcal{E},Sim_{0}}$ $={\bf Real}_{\mathcal{P},\mathcal{E},\mathcal{A}}$. $Sim_{n'}$ runs $\mathcal{S}$ in the ideal-world experiment, so {\bf Hybrid}$_{\mathcal{E},Sim_{n'}}$ $={\bf Ideal}_{\mathcal{F},\mathcal{E},\mathcal{S}}$. 

Therefore,
\begin{equation*}
\begin{array}{ll}
\left|{\bf Real}_{\mathcal{P},\mathcal{E},\mathcal{A}}-{\bf Ideal}_{\mathcal{F},\mathcal{E},\mathcal{S}}\right| & \leq \left|{\bf Hybrid}_{\mathcal{E},Sim_{0}}-{\bf Hybrid}_{\mathcal{E},Sim_{1}}\right|+\left|{\bf Hybrid}_{\mathcal{E},Sim_{1}}-{\bf Hybrid}_{\mathcal{E},Sim_{2}}\right|\\
&+\cdots+\left|{\bf Hybrid}_{\mathcal{E},Sim_{n'-1}}-{\bf Hybrid}_{\mathcal{E},Sim_{n'}}\right|.
\end{array}
\end{equation*}

\begin{lemma} {\bf (Users' Privacy)} For all environments $\mathcal{E}$ and all real world adversaries $\mathcal{A}$ who controls the service provider, there exists an ideal-world simulator $\mathcal{S}$ such that 
\begin{equation*}
\left| {\bf Real}_{\mathcal{P},\mathcal{E},\mathcal{A}}-{\bf Ideal}_{\mathcal{F},\mathcal{E},\mathcal{S}}\right|\leq \frac{1}{2^{\ell}}.
\end{equation*}\label{l1}
\end{lemma}

\begin{proof} Given a real cheating service provider, we can construct a simulator $\mathcal{S}$ in the ideal world experiment such that for any $\mathcal{E}$ cannot distinguish ${\bf Real}_{\mathcal{P},\mathcal{E},\mathcal{A}}$ and ${\bf Ideal}_{\mathcal{F},\mathcal{E},\mathcal{S}}$.
\medskip

\noindent{\bf Game}$_{0}$: $Sim_{0}$ runs $\mathcal{A}$ and the honest user as in the real-world experiment, hence 
\begin{equation*}
{\bf Real}_{\mathcal{P},\mathcal{E},\mathcal{A}}={\bf Hybrid}_{\mathcal{E},Sim_{0}}.
\end{equation*}

\noindent{\bf Game}$_{1}$: $Sim_{1}$ runs the extractor for the proof of knowledge  $\prod_{SP}^{1}: \mbox{PoK}\left\{(\mathfrak{h}):H=e(\mathfrak{g},\mathfrak{h})\right\}$ to extract the knowledge $\mathfrak{h}$ at the first service transfer query dictated by $\mathcal{A}$. If the extractor fails to exact $\mathfrak{h}$, $Sim_{1}$ returns $\perp$ to $\mathcal{E}$; otherwise, $Sim_{1}$ runs $\mathcal{A}$ interacting with $\mathcal{U}$. The difference between ${\bf Hybrid}_{\mathcal{E},Sim_{1}}$ and ${\bf Hybrid}_{\mathcal{E},Sim_{0}}$ is the knowledge error of the proof of knowledge $\prod_{SP}^{1}$. Hence,

\begin{equation*}
\left|{\bf Hybrid}_{\mathcal{E},Sim_{0}}-{\bf Hybrid}_{\mathcal{E},Sim_{1}}\right|\leq \frac{1}{2^{\ell}}.
\end{equation*}

\noindent{\bf Game}$_{2}$: $Sim_{2}$ runs exactly as $Sim_{1}$ in {\bf Game}$_{1}$, except it can retrieve all messages holden by $\mathcal{SP}$. $Sim_{2}$ runs $\mathcal{A}$ to obtain the encrypted $\mathcal{D}'=\left\{((A_{i,j},B_{i,j})_{i=1}^{m})_{j=1}^{n}\right\}$. $Sim_{2}$ can computes $M_{i,j}=\frac{B_{i,j}}{e(\mathfrak{h},A_{i,j})}$ and $\mathcal{D}=\left\{M_{i,j}\right\}$ where $i=1,2,\cdots,m$ and $j=1,2,\cdots,n$. Hence, 
\begin{equation*}
{\bf Hybrid}_{\mathcal{E},Sim_{1}}={\bf Hybrid}_{\mathcal{E},Sim_{2}}.
\end{equation*}

\noindent{\bf Game}$_{3}$: We construct a simulator $\mathcal{S}$ that plays the role as $\mathcal{A}$ in {\bf Game}$_{2}$. $\mathcal{S}$ only relays the communications between $\mathcal{E}$ and $\mathcal{A}$. When receiving a message $(sid, service\_provider,\cdots)$, $\mathcal{S}$ returns $\mathcal{D}$ to $\mathcal{E}$. When receiving a message $(sid,user,O,S)$, $\mathcal{S}$ first checks whether $(O;S)\in \mathfrak{L}$. If it is not, $\mathcal{S}$ returns $(sid,0)$ to $\mathcal{E}$; otherwise, $\mathcal{S}$ returns $(sid,1)$ to $\mathcal{E}$. Hence, 
\begin{equation*}
{\bf Hybrid}_{\mathcal{E},Sim_{2}}={\bf Hybrid}_{\mathcal{E},Sim_{3}}={\bf Ideal}_{\mathcal{F},\mathcal{E},\mathcal{S}}.
\end{equation*}

Therefore,
\begin{equation*}
\begin{array}{c}
\left|{\bf Real}_{\mathcal{P},\mathcal{E},\mathcal{A}}-{\bf Ideal}_{\mathcal{F},\mathcal{E},\mathcal{S}}\right|
\leq \left|{\bf Hybrid}_{\mathcal{E},Sim_{0}}-{\bf Hybrid}_{\mathcal{E},Sim_{1}}\right|+\left|{\bf Hybrid}_{\mathcal{E},Sim_{1}}-{\bf Hybrid}_{\mathcal{E},Sim_{2}}\right|\\
+\left|{\bf Hybrid}_{\mathcal{E},Sim_{2}}-{\bf Hybrid}_{\mathcal{E},Sim_{3}}\right|\leq \frac{1}{2^{\ell}}.
\end{array}
\end{equation*}
\qed
\end{proof}

\begin{lemma} {\bf (Service Provider's Security)} For all environments $\mathcal{E}$ and all real world adversaries $\mathcal{A}$ who controls the user, there exists an ideal-world simulator $\mathcal{S}$ such that 
\begin{equation*}
\left| {\bf Real}_{\mathcal{P},\mathcal{E},\mathcal{A}}-{\bf Ideal}_{\mathcal{F},\mathcal{E},\mathcal{S}}\right|\leq \frac{1}{p}+2Adv_{\mathcal{A}}^{\mbox{q-SDH}}+Adv_{\mathcal{A}}^{\mbox{A-q-PDDE}}.
\end{equation*}\label{l1}
\end{lemma}

\begin{proof} Given a real cheating user, we can construct a simulator $\mathcal{S}$ in the ideal world experiment such at for any $\mathcal{E}$ cannot distinguish ${\bf Real}_{\mathcal{P},\mathcal{E},\mathcal{A}}$ and ${\bf Idea}_{\mathcal{F},\mathcal{E},\mathcal{S}}$.
\medskip

\noindent{\bf Game}$_{0}$: $Sim_{0}$ runs $\mathcal{A}$ and the honest service provider as in the real world experiment, hence,
\begin{equation*}
{\bf Real}_{\mathcal{P},\mathcal{E},\mathcal{A}}={\bf Hybrid}_{\mathcal{E},Sim_{0}}.
\end{equation*}

\noindent{\bf Game}$_{1}$: $Sim_{1}$ runs exactly as $Sim_{0}$ in {\bf Game}$_{0}$, except that $Sim_{1}$ extract the knowledge  $(i,j,r_{1},r_{2},$ $r_{3},r_{4},r_{5},r_{6},r_{7},r_{8},r_{9},r_{10},C_{i,1},C_{i,2},D_{j,1},D_{j,2},\Gamma_{1}^{i},\Gamma_{2}^{j},\Gamma_{1}^{i+l},\Gamma_{2}^{j+k})$ from the proof $\prod_{U}$. $Sim_{1}$ first generates a simulated proof of $\prod_{SP}^{1}: \mbox{PoK}\left\{(\mathfrak{h}): H=e(\mathfrak{g},\mathfrak{h})\right\}$, and then runs the extractor of the knowledge proof of $\prod_{U}$ to extract $(i,j,r_{1},r_{2},r_{3},r_{4},r_{5},r_{6},r_{7},r_{8},r_{9},r_{10},C_{i,1},C_{i,2},D_{j,1},D_{j,2},\Gamma_{1}^{i},\Gamma_{2}^{j},$ $\Gamma_{1}^{i+l},\Gamma_{2}^{j+k})$. Due to the knowledge proof of $\prod_{U}$ is perfect zero-knowledge, we have 
\begin{equation*}
\left|{\bf Hybride}_{\mathcal{E},Sim_{0}}-{\bf Hybrid}_{\mathcal{E},Sim_{1}}\right|\leq \frac{1}{p}
\end{equation*}

\noindent{\bf Game}$_{2}$: $Sim_{2}$ runs exactly as $Sim_{1}$ in {\bf Game}$_{1}$, except that: (1)$ i\notin\{1,2,\cdots,m\}$ or $i+l\notin\{1,2,\cdots,m\}$; $j\notin\{1,2,\cdots,n\}$ or $j+k\notin\{1,2,\cdots,n\}$.
\begin{clm}\label{clm:1}
If the $q$-SDH assumption hold on $(e,p,\mathbb{G},\mathbb{G}_{\tau})$, we have 
\begin{equation*}
\left|{\bf Hybrid}_{\mathcal{E},Sim_{1}}-{\bf Hybrid}_{\mathcal{E},Sim_{2}}\right|\leq 2 Adv_{\mathcal{A}}^{q-SDH}
\end{equation*}
where $q=max\{m+1,n+1\}$.
\end{clm}
\medskip

\noindent{\bf Game}$_{3}:$ $Sim_{3}$ runs exactly as $Sim_{2}$ in {\bf Game}$_{2}$, except that $Sim_{3}$ outputs $(A_{\mu,\nu},L_{\mu,\nu})$ and the proof $\prod_{SP}^{2}$. $Sim_{3}$ computes $A_{\mu,\nu}=\mathfrak{g}^{-(r_{1}+r_{2})}(\mathfrak{g}^{x^{\mu}})^{r_{3}}(\mathfrak{g}^{y^{\nu}})^{r_{5}}g_{1}^{i+\mu}h_{1}^{j+\nu}g_{2}^{x^{i+\nu}}h_{2}^{y^{j+\nu}}$ and $L_{\mu,\nu}=H^{r_{1}+r_{2}}\cdot (H^{x^{\mu}})^{r_{3}}(H^{y^{\nu}})^{r_{5}}\cdot \frac{B_{i+\mu,j+\nu}}{M_{i+\mu,j+\nu}}$, and generates a simulated proof of $\prod_{SP}^{2}=\mbox{PoK}\Big\{(x^{\mu},y^{\nu},\mathfrak{h}):\big((\frac{K_{\mu,\nu}}{E_{1}g_{1}^{\mu}E_{2}h_{1}^{\nu}}=F_{1}^{x^{\mu}}H_{1}^{y^{\nu}}\wedge\frac{e(C_{\mu,2},W_{1})}{e(g_{2},g_{2})}=e(C_{\mu,2},g_{2})^{-x^{\mu}}\wedge$$~\frac{e(D_{\nu,2},W'_{1})}{e(h_{2},h_{2})}=e(D_{\nu,2},h_{2})^{-y^{\nu}}~\wedge~L_{\mu,\nu}=e(\mathfrak{h},K_{\mu,\nu}))_{\mu=0}^{l}\big)_{\nu=0}^{k}\wedge H=e(\mathfrak{g},\mathfrak{h})\Big\}$. Due to the perfect of the zero-knowlege proof, we have that 
\begin{equation*}
{\bf Hybrid}_{\mathcal{E},Sim_{2}}={\bf Hybrid}_{\mathcal{E},Sim_{3}}.
\end{equation*}

\noindent${\bf Game}_{4}:$ $Sim_{4}$ runs  exactly as $Sim_{3}$ in {\bf Game}$_{3}$, except that the values  $(B_{1,1},B_{1,2}, \cdots, B_{m,n})$ are replaced by random elements in $\mathbb{G}_{\tau}$. In this case, the proof $\prod_{SP}^{2}$  in {\bf Game}$_{3}$ is a simulated proof of a false statement. 
\begin{clm}\label{clm:2}
If the $q$-PDDH assumption holds on $(e,p,\mathbb{G},\mathbb{G}_{\tau})$, we have that 
\begin{equation*}
\left|{\bf Hybrid}_{\mathcal{E},Sim_{3}}-{\bf Hybrid}_{\mathcal{E},Sim_{4}}\right|\leq Adv_{\mathcal{A}}^{q-PDDH}
\end{equation*}
where $q=max\{m^{2},n^{2}\}$.
\end{clm}
\medskip

\noindent{\bf Game}$_{5}:$ We construct a simulator $\mathcal{S}$ that works as $\mathcal{A}$ in {\bf Game}$_{4}$. $\mathcal{S}$ only forward the communication between $\mathcal{E}$ and $\mathcal{A}$. When receiving a message $(sid,service\_provider,\mathcal{D})$, stores $M_{i,j}\in\mathcal{D}$ for $i=0,1,2,\cdots,m$ and $j=0,1,2,\cdots,n$.  Upon receiving a message $(sid,user,O,S)$, $\mathcal{S}$ runs the extractor of the proof $\prod_{U}$ to extract $(i,j,r_{1},r_{2},r_{3},r_{4},r_{5},r_{6},r_{7},r_{8},r_{9},r_{10},C_{i,1},C_{i,2},D_{j,1},D_{j,2},\Gamma_{1}^{i},\Gamma_{2}^{j},\Gamma_{1}^{i+l},$ $\Gamma_{2}^{j+k})$. If the extraction fail, $\mathcal{S}$ sends noting to $\mathcal{U}$, otherwise, sends $(sid,service\_request)$ to $\mathcal{SP}$. If $b=0$, returns $(sid,\perp)$ to $\mathcal{U}$. If $b=1$, $\mathcal{S}$ computes $A_{\mu,\nu}=\mathfrak{g}^{-(r_{1}+r_{2})}(\mathfrak{g}^{x^{\mu}})^{r_{5}}(\mathfrak{g}^{y^{\nu}})^{r_{5}}g_{1}^{i+\mu}h_{1}^{j+\nu}g_{2}^{x^{i+\nu}}h_{2}^{y^{j+\nu}}$ and $L_{\mu,\nu}=H^{-(r_{1}+r_{2}}\cdot (H^{x^{\mu}})^{r_{3}}(H^{y^{\nu}})^{r_{5}}\cdot \frac{B_{i+\mu,j+\nu}}{M_{i+\mu,j+\nu}}$, and generates a simulated proof. Hence, 
\begin{equation*}
{\bf Hybrid}_{\mathcal{E},Sim_{4}}={\bf Hybrid}_{\mathcal{E},Sim_{5}}={\bf Ideal}_{\mathcal{F},\mathcal{E},\mathcal{S}}.
\end{equation*}

Therefore,
\begin{equation*}
\begin{array}{ll}
\left|{\bf Real}_{\mathcal{F},\mathcal{E},\mathcal{A}}-{\bf Ideal}_{\mathcal{F},\mathcal{E},\mathcal{S}}\right| & \leq \left|{\bf Hybrid}_{\mathcal{E},Sim_{0}}-{\bf Hybrid}_{\mathcal{E},Sim_{1}}\right|+\left|{\bf Hybrid}_{\mathcal{E},Sim_{1}}-{\bf Hybrid}_{\mathcal{E},Sim_{2}}\right|\\
&+\left|{\bf Hybrid}_{\mathcal{E},Sim_{2}}-{\bf Hybrid}_{\mathcal{E},Sim_{3}}\right|+\left|{\bf Hybrid}_{\mathcal{E},Sim_{3}}-{\bf Hybrid}_{\mathcal{E},Sim_{4}}\right|\\
&+\left|{\bf Hybrid}_{\mathcal{E},Sim_{4}}-{\bf Hybrid}_{\mathcal{E},Sim_{5}}\right|\leq \epsilon(\ell).
\end{array}
\end{equation*}
\qed
\end{proof}

\noindent{\bf Proof of  Claim \ref{clm:1}.} We prove this claim by constructing an algorithm $\mathcal{B}$ that can break the unfogeability under weak chosen-message attack of the Boneh-Boyen signature scheme. According to the proof given in \cite{bb:sig}, $\mathcal{B}$ can solve the $q$-SDH assumption.

Suppose that there exists an environment $\mathcal{E}$ that can distinguish {\bf Game}$_{1}$ and {\bf Game}$_{2}$, $\mathcal{B}$ can forge a signature as follows. We consider the following four cases:  {\em Case-I.} $\mathcal{B}$ outputs a forged signature for $i$ or $i+l$;  {\em Case-II.} $\mathcal{B}$ outputs a forged signature for $j$ or $j+k$.
\medskip

\noindent{\em Case-I.} Given $(g,g^{\zeta},g^{\zeta^{2}},\cdots,g^{\zeta^{q}}, h,h^{\zeta})$, $\mathcal{B}$ sets $g_{1}=g$ and $\mathfrak{h}=h$. $\mathcal{B}$ selects $\gamma_{1},\gamma_{2},\gamma_{3},\gamma_{4}\stackrel{R}{\leftarrow}\mathbb{Z}_{p}$  and computes  $g_{2}=g_{1}^{\gamma_{1}}$,  $h_{1}=g_{1}^{\gamma_{2}}$  $h_{2}=g_{1}^{\gamma_{3}}$,   $\mathfrak{g}=g_{1}^{\alpha_{4}}$,  and sets  $f(\zeta)=(\zeta+1)(\zeta+2)\cdots (\zeta+m)=\sum_{z=0}^{m}a_{z}\zeta^{z}$, $g_{1}=g^{f(\zeta)}$,  and $f_{i}(\zeta)=\frac{f(\zeta)}{\zeta+i}=\sum_{w=0}^{m-1}b_{w}\zeta^{w}$ where $a_{z},b_{w}\in\mathbb{Z}_{p}$ and $i=1,2,\cdots,m$. $\mathcal{B}$ selects $\alpha_{2},\beta_{1},\beta_{2},x,y\stackrel{\$}{\leftarrow}\mathbb{Z}_{p}$, and computes $H=e(\mathfrak{g},\mathfrak{h})$ and sets $\alpha_{1}=\zeta$. $\mathcal{B}$ computes
\begin{equation*}
\begin{array}{c}
W_{1}=\prod_{k=0}^{m-1} (g^{\zeta^{k+1}})^{a_{z}}=(g^{\sum_{z=0}^{m}a_{z}\zeta^{z}})^{\zeta}=(g^{f(\zeta))\zeta}=g_{1}^{\zeta},~ W_{2}=g_{2}^{\alpha_{2}},  ~ W_{1}'=h_{2}^{\beta_{1}}, ~ W_{2}'=h_{1}^{\beta_{2}}, \\  \Gamma_{1}^{i}=\prod_{w=0}^{m-1}(g^{\zeta})^{b_{w}}=g^{\sum_{w=0}^{m-1}b_{w}\zeta^{w}}= g^{f_{i}(\zeta)}=g^{\frac{f(\zeta)}{\zeta+i}}=g_{1}^{\frac{1}{\alpha_{1}+i}}, ~ \Gamma_{2}^{j}=h_{1}^{\frac{1}{\beta_{1}+j}}, ~A_{i,j}=(g_{1}^{i}h_{1}^{j}g_{2}^{x^{i}}h_{2}^{y^{j}}), \\
 B_{i,j}=e(\mathfrak{h},A_{i,j})\cdot M_{i,j}, ~  \left(C_{i,1}=g_{2}^{x^{i}}, ~ C_{i,2}=g_{2}^{\frac{1}{\alpha_{2}+x^{i}}}, ~ C_{i,3}=e(\mathfrak{h},\mathfrak{g})^{x^{i}}\right), \\
 \left(D_{j,1}=h_{2}^{y^{j}}, ~ D_{j,2}=h_{2}^{\frac{1}{\beta_{2}+y^{j}}}, ~ D_{j,3}=e(\mathfrak{h},\mathfrak{g})^{y^{j}}\right) \mbox{ for}~  i=1,2,\cdots,m ~\mbox{and}~ j=1,2,\cdots,n.
\end{array}
 \end{equation*}
 
 The secret key is $SK=(\zeta,\alpha_{2},\beta_{1},\beta_{2},x,y,\mathfrak{h})$ and  the public parameters  are 
$PP=\big(e,p,\mathbb{G},\mathbb{G}_{\tau},$ $\mathfrak{g},g_{1},g_{2},h_{1},h_{2},H,W_{1},W_{2},W'_{1},W'_{2},\Gamma_{1}^{1},\cdots,\Gamma_{1}^{m},\Gamma_{2}^{1},\cdots,\Gamma_{2}^{n},((A_{1,1},B_{1,1}),\cdots,(A_{m,n},B_{m,n}),(C_{1,1},$ $C_{1,2},C_{1,3}),\cdots,(C_{m,1},C_{m,2},C_{m,3}),(D_{1,1},D_{2,1},D_{1,3}),\cdots,(D_{n,1},D_{n,2},D_{n,3})\big)$ and $\mathcal{D}'=\{((A_{i,j},$\\ $ B_{i,j})_{i=1}^{m})_{j=1}^{n}\}$.
\medskip

$\mathcal{B}$ runs the extractor of the proof  $\prod_{U}$ to extract the knowledge $(i,j,r_{1},r_{2},r_{3},r_{4},r_{5},r_{6},r_{7},r_{8},$ $C_{i,1},C_{i,2},\Gamma_{1}^{i},\Gamma_{2}^{j},\Gamma_{1}^{i+l},\Gamma_{2}^{j+k})$. If $\mathcal{E}$ can distinguish {\bf Game}$_{1}$ and {\bf Game}$_{2}$, namely $i\notin\{1,2,\cdots,m\}$ or $i+l\notin\{1,2,\cdots,m\}$, $\mathcal{B}$ outputs a forged signature $I_{1}^{\frac{1}{r_{7}}}$ on $i$ or a forged signature $(I_{3})^{\frac{1}{r_{9}}}$ on $i+l$.
\medskip

\noindent{\em Case-II.} Given $(g,g^{\zeta},g^{\zeta^{2}},\cdots,g^{\zeta^{q}})$, $\mathcal{B}$ selects $g_{1}$, $g_{2}$  $h_{2}$,    $\mathfrak{g}$ and $\mathfrak{h}$  from $\mathbb{G}$, and sets  $f(\zeta)=(\zeta+1)(\zeta+2)\cdots (\zeta+n)=\sum_{z=0}^{n}c_{z}\zeta^{z}$, $h_{1}=g^{f(\zeta)}$,  and $f_{j}(\zeta)=\frac{f(\zeta)}{\zeta+i}=\sum_{w=1}^{n-1}d_{w}\zeta^{w}$ where $c_{z},d_{w}\in\mathbb{Z}_{p}$ and $j=1,2,\cdots,n$. $\mathcal{B}$ selects $\alpha_{1},\alpha_{2},\beta_{2},x,y\stackrel{\$}{\leftarrow}\mathbb{Z}_{p}$, and computes $H=e(\mathfrak{g},\mathfrak{h})$ and sets $\beta_{1}=\zeta$. $\mathcal{B}$ computes
\begin{equation*}
\begin{array}{c}
 W_{1}=g_{1}^{\alpha_{1}}, ~ W_{2}=g_{2}^{\alpha_{2}}, ~ W_{1}'=\prod_{z=0}^{n} (g^{\zeta^{z+1}})^{c_{z}}=(g^{\sum_{z=0}^{n}c_{z}\zeta^{z}})^{\zeta}=(g^{f(\zeta))\zeta}=h_{1}^{\zeta}, ~ W_{2}'=h_{2}^{\beta_{2}}, \\
  ~ \Gamma_{1}^{j}=g_{1}^{\frac{1}{\alpha_{1}+i}},~  \Gamma_{2}^{j}=\prod_{w=0}^{n-1}(g^{\zeta^{w}})^{d_{w}}=g^{\sum_{w=0}^{n-1}d_{w}\zeta^{w}}= g^{f_{j}(\zeta)}=g^{\frac{f(\zeta)}{\zeta+j}}=h_{1}^{\frac{1}{\beta_{1}+j}}, ~A_{i,j}=(g_{1}^{i}h_{1}^{j}g_{2}^{x^{i}}h_{2}^{y^{j}}),\\ 
 B_{i,j}=e(\mathfrak{h},A_{i,j})\cdot M_{i,j}, ~  \left(C_{i,1}=g_{2}^{x^{i}}, ~ C_{i,2}=g_{2}^{\frac{1}{\alpha_{2}+x^{i}}}, ~ C_{i,3}=e(\mathfrak{h},\mathfrak{g})^{x^{i}}\right), \\
 \left(D_{j,1}=h_{2}^{y^{j}}, ~ D_{j,2}=h_{2}^{\frac{1}{\beta_{2}+y^{j}}}, ~ D_{j,3}=e(\mathfrak{h},\mathfrak{g})^{y^{j}}\right) \mbox{ for}~  i=1,2,\cdots,m ~\mbox{and}~ j=1,2,\cdots,n.
\end{array}
 \end{equation*}
 
 The secret key is $SK=(\alpha_{1},\alpha_{2},\zeta,\beta_{2},x,y,\mathfrak{h})$ and  the public parameters  are 
$PP=\big(e,p,\mathbb{G},\mathbb{G}_{\tau},$ $\mathfrak{g},g_{1},g_{2},h_{1},h_{2},H,W_{1},W_{2},W'_{1},W'_{2},\Gamma_{1}^{1},\cdots,\Gamma_{1}^{m},\Gamma_{2}^{1},\cdots,\Gamma_{2}^{n},((A_{1,1},B_{1,1}),\cdots,(A_{m,n},B_{m,n}),(C_{1,1},$ $C_{1,2},C_{1,3}),\cdots,(C_{m,1},C_{m,2},C_{m,3}),(D_{1,1},D_{2,1},D_{1,3}),\cdots,(D_{n,1},D_{n,2},D_{n,3})\big)$ and $\mathcal{D}'=\{((A_{i,j},$\\ $ B_{i,j})_{i=1}^{m})_{j=1}^{n}\}$.
\medskip

$\mathcal{B}$ runs the extractor of the proof  $\prod_{U}$ to extract the knowledge $(i,j,r_{1},r_{2},r_{3},r_{4},r_{5},r_{6},r_{7},r_{8},r_{9},$ $r_{10},C_{i,1},C_{i,2},D_{j,1},D_{j,2},\Gamma_{1}^{i},\Gamma_{2}^{j},\Gamma_{1}^{i+l},\Gamma_{2}^{j+k})$. If $\mathcal{E}$ can distinguish {\bf Game}$_{1}$ and {\bf Game}$_{2}$, namely $j\notin\{1,2,\cdots,n\}$ or $j+k\notin\{1,2,\cdots,m\}$, $\mathcal{B}$ outputs a forged signature $I_{2}^{\frac{1}{r_{8}}}$ on $j$ or a forged signature $(I_{4})^{\frac{1}{r_{10}}}$ on $j+k$.
\medskip

Therefore, 

\begin{equation*}
\left|{\bf Hybrid}_{\mathcal{E},Sim_{2}}-{\bf Hybrid}_{\mathcal{E},Sim_{1}}\right|\leq 2 Adv_{\mathcal{A}}^{q-SDH}.
\end{equation*}
\qed
\medskip

\noindent{\bf Proof of Claim \ref{clm:2}.} We prove this claim by constructing an algorithm $\mathcal{B}$ that can break the  $q$-PDDH assumption.

Suppose that there exists an environment $\mathcal{E}$ that can distinguish {\bf Game}$_{3}$ and {\bf Game}$_{4}$, $\mathcal{B}$ can break the $q$-PDDH as follows.

Given $(g,g^{\zeta},g^{\zeta^{2}},\cdots,g^{\zeta^{q}},H,T_{1},T_{2},\cdots,T_{q})$, $\mathcal{B}$ will determine whether $T_{z}=H^{x^{z}}$ or   $T_{z}\stackrel{R}{\leftarrow}\mathbb{Z}_{p}$ for $z=1,2,\cdots,q$. Let $f(x)=(\alpha_{2}+x)(\alpha_{2}+x^{2})\cdots(\alpha_{2}+x^{m})=\sum_{z=0}^{\frac{m(1+m)}{2}}a_{z}x^{z}$, $f_{i}(x)=(\alpha_{2}+x)(\alpha_{2}+x^{2})\cdots(\alpha_{2}+x^{i-1})(\alpha_{2}+x^{i+1})\cdots(\alpha_{2}+x^{m})=\sum_{w=0}^{\frac{m(1+m)}{2}-i}b_{w}x^{w}$, $f'(y)=(\beta_{2}+y)(\beta_{2}+y^{2})\cdots(\beta_{2}+y^{n})=\sum_{\rho=0}^{\frac{n(1+n)}{2}}c_{\rho}y^{\rho}$ and $f'_{i}(y)=(\beta_{2}+y)(\beta_{2}+y^{2})\cdots(\beta_{2}+y^{j-1})(\beta_{2}+y^{j+1})\cdots(\beta_{2}+y^{n})=\sum_{\varrho=0}^{\frac{n(1+n)}{2}-j}d_{\varrho}y^{\varrho}$. 
$\mathcal{B}$ selects $\gamma_{1},\gamma_{2}\stackrel{R}{\leftarrow}\mathbb{Z}_{p}$, and sets $\mathfrak{g}=g$, $g_{1}=g^{\gamma_{1}}$, $g_{2}=g^{f(x)}$, $h_{1}=g^{\gamma_{2}}$ and $h_{2}=g^{f'(y)}$. $\mathcal{B}$ selects $\alpha_{1},\alpha_{2},\beta_{1},\beta_{2},\gamma\stackrel{R}{\leftarrow}\mathbb{Z}_{p}$, and sets $y=\gamma x$.        $\mathcal{B}$ computes 
\begin{equation*}
\begin{array}{c}
W_{1}=g_{1}^{\alpha_{1}}, ~W_{2}=g_{2}^{\alpha_{2}}, ~ W'_{1}=h_{1}^{\beta_{1}},~W'_{2}=h_{2}^{\beta_{2}},~
\Gamma_{1}^{i}=g_{1}^{\frac{1}{\alpha_{1}+i}},~\Gamma_{2}^{j}=h_{1}^{\frac{1}{\beta_{1}+j}},\\
\end{array}
\end{equation*}
\begin{equation*}
\begin{array}{ll}
A_{i,j}&=g_{1}^{i}h_{1}^{j}\prod_{z=0}^{\frac{m(1+m)}{2}}(g^{x^{z+i}})^{a_{z}}\prod_{\rho=0}^{\frac{n(1+n)}{2}}(g^{x^{\rho+j}})^{\gamma^{\rho+j}c_{\rho}}\\
&=g_{1}^{i}h_{1}^{j}\prod_{z=0}^{\frac{m(1+m)}{2}}(g^{a_{z}x^{z}})^{x^{i}}\prod_{\rho=0}^{\frac{n(1+n)}{2}}(g^{c_{z}y^{\rho}})^{y_{j}}=g_{1}^{i}h_{1}^{j}g_{2}^{x^{i}}h_{2}^{y^{j}},\\
\end{array}
\end{equation*}
\begin{equation*}
\begin{array}{ll}
B_{i,j}=H^{\gamma_{1}i+\gamma_{2}j}\cdot \prod_{z=0}^{\frac{m(1+m)}{2}}T_{z+i}^{a_{z}}\prod_{\rho=0}^{\frac{n(1+n)}{2}}T_{\rho+j}^{c_{\rho}}\cdot M_{i,j},\\
C_{i,1}=\prod_{z=0}^{\frac{m(1+m)}{2}}(g^{x^{z+i}})^{a_{z}}=\prod_{z=0}^{m}(g^{a_{z}x^{z}})^{x^{i}}=g_{2}^{x^{i}},\\
C_{i,2}=\prod_{w=0}^{\frac{m(m+1)}{2}-i}(g^{x^{w}})^{b_{w}}=\prod_{w=0}^{\frac{m(m+1)}{2}-i}(g^{b_{w}x^{w}})=g^{f_{i}(x)}
=g^{\frac{f(x)}{\alpha_{2}+x^{i}}}=g_{2}^{\frac{1}{\alpha_{2}+x^{i}}},\\
C_{i,3}=T_{i},\\
D_{j,1}=\prod_{\rho=0}^{\frac{n(1+n)}{2}}(g^{x^{\rho+j}})^{\gamma^{\rho+j}c_{\rho}}=\prod_{\rho=0}^{\frac{n(1+n)}{2}}(g^{(\gamma x)^{\rho+j}})^{c_{\rho}}=\prod_{\rho=0}^{\frac{n(1+n)}{2}}(g^{c_{\rho}y^{\rho}})^{y^{j}}=h_{2}^{y^{j}},\\
D_{j,2}=\prod_{\varrho=0}^{\frac{n(1+n)}{2}}(g^{x^{\varrho}})^{d_{\varrho}\gamma^{\varrho}}=\prod_{\varrho=0}^{\frac{n(1+n)}{2}}(g^{(\gamma x)^{\varrho}})^{d_{\varrho}}=\prod_{\varrho=0}^{\frac{n(1+n)}{2}}(g^{d_{\varrho}y^{\varrho}})=h_{2}^{\frac{1}{\beta_{2}+y^{j}}},\\
D_{j,3}=T_{j}^{\gamma^{j}},~
\mbox{for}~ i=1,2,\cdots,m ~\mbox{and}~ j=1,2,\cdots,n.
\end{array}
\end{equation*}

The secret key is $SK=(\alpha_{1},\alpha_{2},\beta_{1},\beta_{2},x,y)$ and the public parameters $PP=\big(e,p,\mathbb{G},\mathbb{G}_{\tau},$ $\mathfrak{g},g_{1},g_{2},h_{1},h_{2},H,W_{1},W_{2},W'_{1},W'_{2},\Gamma_{1}^{1},\cdots,\Gamma_{1}^{m},\Gamma_{2}^{1},\cdots,\Gamma_{2}^{n},((A_{1,1},B_{1,1}),\cdots,(A_{m,n},B_{m,n}),(C_{1,1},$ $C_{1,2},C_{1,3}),\cdots,(C_{m,1},C_{m,2},C_{m,3}),(D_{1,1},D_{2,1},D_{1,3}),\cdots,(D_{n,1},D_{n,2},D_{n,3})\big)$ and $\mathcal{D}'=\{((A_{i,j},$\\ $ B_{i,j})_{i=1}^{m})_{j=1}^{n}\}$. $\mathcal{B}$ sends $PP$ to $\mathcal{E}$.

If $(T_{1},T_{2},\cdots,T_{q})=(H^{x},H^{x^{2}},\cdots,H^{x^{q}})$, the parameters are distributed exactly as in {\bf Game}$_{3}$. If $(T_{1},T_{2},\cdots,T_{q})\stackrel{R}{\leftarrow}\mathbb{G}_{\tau}^{q}$, the parameters are distributed exactly as in {\bf Game}$_{4}$. Hence, $\mathcal{B}$ can break the $q$-PDDH assumption if $\mathcal{E}$ can distinguish {\bf Game}$_{3}$ from {\bf Game}$_{4}$. Therefore, we have
\begin{equation*}
\left|{\bf Hybrid}_{\mathcal{E},Sim_{3}}-{\bf Hybrid}_{\mathcal{E},Sim_{4}}\right|\leq Adv_{\mathcal{A}}^{q-PDDH}.
\end{equation*}

\section{Conclusion and Future Work}\label{sec:conc}
In this paper, we proposed an OLBSQ scheme which does not require a semi-TTP. Especially, in our OLBSQ scheme, both the computation cost and communication cost to generate a query is constant, instead of linear with the  size of the queried area. We formalised the definition and security model of our OLBSQ scheme, and presented a concrete construction. 
Finally, we reduced the security of the proposed OLBSQ scheme to well-known complexity assumptions.  

Our OLBSQ scheme was constructed on the groups equipped with pairing. Comparatively, pairing is time consuming operation. Therefore, constructing  OLBSQ schemes without pairing is interesting and desirable. We leave it as an open problem and our future work. 

\bibliographystyle{plain}
\bibliography{references}

\appendix
\section{Correctness}
\noindent{\em Correctness.} Our scheme described in Fig. \ref{fig:setup} and Fig. \ref{fig:service_tr} is correct because the following equations hold.
\begin{equation*}
\begin{split}
F_{1}=\mathfrak{g}^{r_{3}}C_{i,1}=\mathfrak{g}^{r_{3}}g_{2}^{x^{i}},~
F_{2}=C_{i,2}^{r_{4}}=g_{2}^{\frac{r_{4}}{\alpha_{2}+x^{i}}},~
J_{1}=\mathfrak{g}^{r_{5}}D_{j,1}=\mathfrak{g}^{r_{5}}h_{2}^{y^{j}},~
J_{2}=D_{j,2}^{r_{6}}=h_{2}^{\frac{r_{6}}{\beta_{2}+y^{j}}},\\
I_{1}=(\Gamma_{1}^{i})^{r_{7}}=g_{1}^{\frac{r_{7}}{\alpha_{1}+i}},~
I_{2}=(\Gamma_{2}^{j})^{r_{8}}=h_{1}^{\frac{r_{8}}{\beta_{1}+j}},~
I_{3}=(\Gamma_{1}^{i+l})^{r_{9}}=g_{1}^{\frac{r_{9}}{\alpha_{1}+i+l}},~
I_{4}=(\Gamma_{2}^{j+k})^{r_{10}}=h_{1}^{\frac{r_{10}}{\beta_{1}+j+k}},
\end{split}
\end{equation*}

\begin{equation*}
\begin{split}
e(I_{1},W_{1}^{-1})=e(g_{1}^{\frac{r_{7}}{\alpha_{1}+i}},g_{1}^{-\alpha_{1}})=e(g_{1}^{\frac{-r_{7}(\alpha_{1}+i)+ir_{7}}{\alpha_{1}+i}},g_{1})=e(g_{1},g_{1})^{-r_{7}}\cdot e(g_{1},I_{1})^{i},\\
e(I_{2},(W_{1}')^{-1})=e(h_{1}^{\frac{r_{8}}{\beta_{1}+j}},h_{1}^{-\beta_{1}})=e(h_{1}^{\frac{-r_{8}(\beta_{1}+i)+jr_{8}}{\beta_{1}+j}},h_{1})=e(h_{1},h_{1})^{-r_{8}}\cdot e(h_{1},I_{2})^{j},\\
e(I_{3},W_{1}^{-1})\times e(g_{1},I_{3})^{-l}=e(g_{1}^{\frac{r_{9}}{\alpha_{1}+i+l}},g_{1}^{-\alpha_{1}})\times e(g_{1},I_{3})^{-l}=e(g_{1}^{\frac{-r_{9}(\alpha_{1}+i+l)+r_{9}(i+l)}{\alpha_{1}+i+l}},g_{1})\times e(g_{1},I_{3})^{-l}\\=e(g_{1},g_{1})^{-r_{9}}\cdot e(g_{1},I_{3})^{(i+l)}\times e(g_{1},I_{3})^{-l}=e(g_{1},g_{1})^{-r_{9}}\cdot e(g_{1},I_{3})^{i},\\
e(I_{3},W_{1}^{-1})\times e(g_{1},I_{3})^{-l)}=e(g_{1},g_{1})^{-r_{9}}\cdot e(g_{1},I_{3})^{i},\\
e(I_{4},(W'_{1})^{-1})\times e(h_{1},I_{4})^{-k}=e(h_{1}^{\frac{r_{10}}{\beta_{1}+j+k}},h_{1}^{-\beta_{1}})\times e(h_{1},I_{4})^{-k}=e(h_{1}^{\frac{-r_{10}(\beta_{1}+j+k)+r_{10}(j+k)}{\beta_{1}+j+k}},h_{1})\times e(h_{1},I_{4})^{-k}\\=e(h_{1},h_{1})^{-r_{10}}\cdot e(h_{1},I_{4})^{(j+k)}\times e(h_{1},I_{4})^{-k}=e(h_{1},h_{1})^{-r_{10}}\cdot e(h_{1},I_{4})^{j},
\end{split}
\end{equation*}

\begin{equation*}
\begin{split}
e(F_{1}W_{2},F_{2})=e(\mathfrak{g}^{r_{3}}g_{2}^{x^{i}}g_{2}^{\alpha_{2}},g_{2}^{\frac{r_{4}}{\alpha_{2}+x^{i}}})=e(\mathfrak{g}^{r_{3}}g_{2}^{\alpha_{2}+x^{i}},g_{2}^{\frac{r_{4}}{\alpha_{2}+x^{i}}})=e(\mathfrak{g},F_{2})^{r_{3}}\cdot e(g_{2},g_{2})^{r_{4}},\\
e(J_{1}W'_{2},J_{2})=e(\mathfrak{g}^{r_{5}}h_{2}^{y^{j}}h_{2}^{\beta_{2}},h_{2}^{\frac{r_{6}}{\beta_{2}+y^{j}}})=e(\mathfrak{g}^{r_{5}}h_{2}^{\beta_{2}+y^{j}},h_{2}^{\frac{r_{6}}{\beta_{2}+y^{j}}})=e(\mathfrak{g},J_{2})^{r_{5}}\cdot e(h_{2},h_{2})^{r_{6}},\\
e(E_{1}W_{1},I_{1})=e(\mathfrak{g}^{r_{1}}g_{1}^{i}g_{1}^{\alpha_{1}},g_{1}^{\frac{r_{7}}{\alpha_{1}+i}})=e(\mathfrak{g}^{r_{1}}g_{1}^{\alpha_{1}+i},g_{1}^{\frac{r_{7}}{\alpha_{1}+i}})=e(\mathfrak{g},I_{1})^{r_{1}}\cdot e(g_{1},g_{1})^{r_{7}},\\
e(E_{2}W'_{1},I_{2})=e(\mathfrak{g}^{r_{2}}h_{1}^{j}h_{1}^{\beta_{1}},h_{1}^{\frac{r_{8}}{\beta_{1}+j}})=e(\mathfrak{g}^{r_{2}}h_{1}^{\beta_{1}+j},h_{1}^{\frac{r_{8}}{\beta_{1}+j}})=e(\mathfrak{g},I_{2})^{r_{2}}\cdot e(h_{1},h_{1})^{r_{8}},\\
e(E_{1}g_{1}^{l}W_{1},I_{3})=e(\mathfrak{g}^{r_{1}}g_{1}^{i}g_{1}^{l}g_{1}^{\beta_{1}},g_{1}^{\frac{r_{9}}{\alpha_{1}+i+l}})=e(\mathfrak{g}^{r_{2}}h_{1}^{\alpha_{1}+i+l},g_{1}^{\frac{r_{9}}{\beta_{1}+j+k}})=e(\mathfrak{g},I_{3})^{r_{1}}\cdot e(g_{1},g_{1})^{r_{9}},\\
e(E_{2}h_{1}^{k}W'_{1},I_{4})=e(\mathfrak{g}^{r_{2}}h_{1}^{j}h_{1}^{k}h_{1}^{\beta_{1}},h_{1}^{\frac{r_{10}}{\beta_{1}+j+k}})=e(\mathfrak{g}^{r_{2}}h_{1}^{\beta_{1}+j+k},h_{1}^{\frac{r_{10}}{\beta_{1}+j+k}})=e(\mathfrak{g},I_{4})^{r_{2}}\cdot e(h_{1},h_{1})^{r_{10}},\\
\frac{e(C_{\mu,2},W_{2})}{e(g_{2},g_{2})}=\frac{e(g_{2}^{\frac{1}{\alpha_{2}+x^{\mu}}},g_{2}^{\alpha_{2}})}{e(g_{2},g_{2})}=\frac{e(g_{2}^{\frac{(\alpha_{2}+x^{\mu})-x^{\mu}}{\alpha_{2}+x^{\mu}}},g_{2})}{e(g_{2},g_{2})}=\frac{e(g_{2},g_{2})\cdot e(C_{\mu,2},g_{2})^{-x^{\mu}}}{e(g_{2},g_{2})}=e(C_{\mu,2},g_{2})^{-x^{\mu}},\\
\frac{e(D_{\nu,2},W'_{2})}{e(h_{2},h_{2})}=\frac{e(h_{2}^{\frac{1}{\beta_{2}+y^{\nu}}},h_{2}^{\beta_{2}})}{e(h_{2},h_{2})}=\frac{e(h_{2}^{\frac{(\beta_{2}+y^{\nu})-y^{\mu}}{\beta_{2}+y^{\nu}}},h_{2})}{e(h_{2},h_{2})}=\frac{e(h_{2},h_{2})\cdot e(D_{\nu,2},h_{2})^{-y^{\nu}}}{e(g_{2},g_{2})}=e(D_{\nu,2},h_{2})^{-y^{\nu}},
\end{split}
\end{equation*}

\begin{equation*}
\begin{split}
K_{\mu,\nu}=E_{1}g_{1}^{\mu}E_{2}h_{1}^{\nu}F_{1}^{x^{\mu}}J_{1}^{y^{\nu}}=\mathfrak{g}^{r_{1}}g_{1}^{i}g_{1}^{\mu}\mathfrak{g}^{r_{2}}h_{1}^{j}h_{1}^{\nu}(\mathfrak{g}^{r_{3}}g_{2}^{x^{i}})^{x^{\mu}}(\mathfrak{g}^{r_{5}}h_{2}^{y^{j}})^{y^{\nu}}\\=\mathfrak{g}^{-(r_{1}+r_{2})+r_{3}x^{\mu}+r_{5}y^{\nu}}g_{1}^{i+\mu}h_{1}^{j+\nu}g_{2}^{x^{i+\mu}}h_{2}^{y^{j+\nu}},\\
L_{\mu,\nu}=e(K_{\mu,\nu},\mathfrak{h})=e(\mathfrak{g},\mathfrak{h})^{-(r_{1}+r_{2})}\cdot e(\mathfrak{g},\mathfrak{h})^{r_{3}x^{\mu}}\cdot e(\mathfrak{g},\mathfrak{h})^{r_{5}y^{\nu}}\cdot e(A_{i+\mu,j+\nu},\mathfrak{h}),\\
P_{\mu,\nu}=\frac{L_{\mu,\nu}}{H^{-(r_{1}+r_{2})}\cdot C_{\mu,3}^{r_{3}}\cdot D_{\nu,3}^{r_{5}}}\\=\frac{e(\mathfrak{g},\mathfrak{h})^{-(r_{1}+r_{2})}\cdot e(\mathfrak{g},\mathfrak{h})^{r_{3}x^{\mu}}\cdot e(\mathfrak{g},\mathfrak{h})^{r_{5}y^{\nu}}\cdot e(A_{i+\mu,j+\nu},\mathfrak{h})}{e(\mathfrak{g},\mathfrak{h})^{-(r_{1}+r_{2})}\cdot e(\mathfrak{g},\mathfrak{h})^{r_{3}x^{\mu}}\cdot e(\mathfrak{g},\mathfrak{h})^{r_{5}y^{\nu}}}=e(A_{i+\mu,j+\nu},\mathfrak{h}),\\
\frac{B_{i+\mu,j+\nu}}{P_{\mu,\nu}}=\frac{e(A_{i+\mu,j+\nu},\mathfrak{h})\cdot M_{i+\mu,j+\nu}}{e(A_{i+\mu,j+\nu},\mathfrak{h})}=M_{i+\mu,j+\nu}.
\end{split}
\end{equation*}

\section*{Details of Zero-Knowledge Proofs}
Let $\mathcal{H}:\{0,1\}^{*}\rightarrow\mathbb{Z}_{p}$ be a cryptographic hash function.
\medskip

\noindent{\em An Instance of Zero Knowledge Proof $\prod_{SP}^{1}$.}

\begin{enumerate}
\item $\mathcal{SP}$ selects $\mathfrak{h}'\stackrel{R}{\leftarrow}\mathbb{G}$ and $M_{SP}^{1}\stackrel{R}{\leftarrow}\{0,1\}^{*}$, and computes $H'=e(\mathfrak{g},\mathfrak{h}')$, $c=\mathcal{H}(H||H'||M_{SP}^{1})$ and $\hat{\mathfrak{h}}=\mathfrak{h}'\mathfrak{h}^{-c}$. $SP$ sends $(H,H',c, \hat{\mathfrak{h}},M_{SP}^{1})$ to $\mathcal{U}$.

\item $\mathcal{U}$ checks $c\stackrel{?}{=}\mathcal{H}(H||H'||M_{SP}^{1})$ and $H'\stackrel{?}{=}e(\mathfrak{g},\hat{\mathfrak{h}})\cdot H^{c}$.
\end{enumerate}
\medskip

\noindent{\em An Instance of Zero Knowledge Proof $\prod_{U}$.}

\begin{enumerate}
\item $\mathcal{U}$ selects $r,s,r_{1},r_{2},r_{3},r_{4},r_{5},r_{6},r_{7},r_{8},r_{9},r_{10},s_{1},s_{2},s_{3},s_{4},s_{5},s_{6},s_{7},s_{8},s_{9},s_{10}\stackrel{R}{\leftarrow}\mathbb{Z}_{p}$, $M_{U}\stackrel{R}{\leftarrow}\{0,1\}^{*}$, and computes
$E_{1}=\mathfrak{g}^{-r_{1}}g_{1}^{i}$, $E_{2}=\mathfrak{g}^{-r_{2}}h_{1}^{j}$,
 $F_{1}=\mathfrak{g}^{r_{3}}C_{i,1}$,  $F_{2}=C_{i,2}^{r_{4}}$,
 $J_{1}=\mathfrak{g}^{r_{5}}D_{j,1}$, $J_{2}=D_{j,2}^{r_{6}}$,
   $I_{1}=(\Gamma_{1}^{i})^{r_{7}}$, $I_{2}=(\Gamma_{2}^{j})^{r_{8}}$,  
   $I_{3}=(\Gamma_{1}^{i+l})^{r_{9}}$, $I_{4}=(\Gamma_{2}^{j+k})^{r_{10}}$,
$E'_{1}=\mathfrak{g}^{s_{1}}g_{1}^{r}$, $E'_{2}=\mathfrak{g}^{s_{2}}h_{1}^{s}$,
$\Theta_{1}=e(\mathfrak{g},F_{2})^{s_{3}}$,  $\Theta_{2}=e(g_{2},g_{2})^{s_{4}}$,
$\Theta_{3}=e(\mathfrak{g},J_{2})^{s_{5}}$, $\Theta_{4}=e(h_{2},h_{2})^{s_{6}}$,
   $\Theta_{5}=e(g_{1},I_{1})^{r}$, $\Theta_{6}=e(g_{1},g_{1})^{s_{7}}$,  
   $\Theta_{7}=e(h_{1},I_{2})^{s}$, $\Theta_{8}=e(h_{1},h_{1})^{s_{8}}$,
   $\Theta_{9}=e(g_{1},g_{1})^{s_{9}}$, $\Theta_{10}=e(h_{1},h_{1})^{s_{10}}$, 
   $\Theta_{11}=e(\mathfrak{g},I_{1})^{s_{1}}$, $\Theta_{12}=e(\frak{g},I_{2})^{s_{2}}$, $\Theta_{13}=e(\frak{g},I_{3})^{s_{1}}$, $\Theta_{14}=e(\frak{g},I_{4})^{s_{2}}$, $\Theta_{15}=e(g_{1},I_{3})^{r}$, $\Theta_{16}=e(h_{1},I_{4})^{s}$,\\

$c_{1}=\mathcal{H}(E_{1}||E_{1}'||M_{U})$, $c_{2}=\mathcal{H}(E_{2}||E'_{2}||M_{U})$, $c_{3}=\mathcal{H}(\Theta_{5}||\Theta_{6}||M_{U})$, $c_{4}=\mathcal{H}(\Theta_{7}||\Theta_{8}||M_{U})$, $c_{5}=\mathcal{H}(\Theta_{9}||\Theta_{15}||M_{U})$, $c_{6}=\mathcal{H}(\Theta_{10}||\Theta_{16}||M_{U})$, $c_{7}=\mathcal{H}(\Theta_{1}||\Theta_{2}||M_{U})$, $c_{8}=\mathcal{H}(\Theta_{3}||\Theta_{4}||M_{U})$, $c_{9}=\mathcal{H}(\Theta_{6}||\Theta_{11}||M_{U})$, $c_{10}=\mathcal{H}(\Theta_{8}||\Theta_{12}||M_{U})$, $c_{11}=\mathcal{H}(\Theta_{9}||\Theta_{13}||M_{U})$, $c_{12}=\mathcal{H}(\Theta_{10}||\Theta_{14}||M_{U})$, \\

$z_{1}=s_{1}+c_{1}r_{1}$, $z_{2}=r-c_{1}i$,  $z_{3}=s_{2}+c_{2}r_{2}$, $s_{4}=s-c_{2}j$, $z_{5}=s_{7}+c_{3}r_{7}$, $z_{6}=r-c_{3}i$, $z_{7}=s-c_{4}j$, $z_{8}=s_{8}+c_{4}r_{8}$, $z_{9}=s_{9}+c_{5}r_{9}$, $z_{10}=r-c_{5}i$, $z_{11}=s_{10}+c_{6}r_{10}$, $z_{12}=s-c_{6}j$, $z_{13}=s_{3}-c_{7}r_{3}$, $z_{14}=s_{4}-c_{7}r_{4}$, $z_{15}=s_{5}-c_{8}r_{5}$, $z_{16}=s_{6}-c_{8}r_{6}$, $z_{17}=s_{7}-c_{9}r_{7}$, $z_{18}=s_{1}+c_{9}r_{1}$, $z_{19}=s_{8}-c_{10}r_{8}$, $z_{20}=s_{2}+c_{10}r_{2}$, $z_{21}=s_{9}-c_{11}r_{9}$, $z_{22}=s_{1}+c_{11}r_{1}$, $z_{23}=s_{10}-c_{12}r_{10}$, $z_{24}=s_{2}+c_{12}r_{2}$\\

$\mathcal{U}$ sends $(E_{1},E_{2},F_{1},F_{2},J_{1},J_{2},I_{1},I_{2},I_{3},I_{4},E'_{1},E'_{2},\Theta_{1},\Theta_{2},\Theta_{3},\Theta_{4},\Theta_{5},\Theta_{6},\Theta_{7},\Theta_{8},\Theta_{9},\Theta_{10},\Theta_{11},$ $\Theta_{12},\Theta_{13},\Theta_{14},$ $\Theta_{15}, \Theta_{16},c_{1},c_{2},c_{3},c_{4},c_{5},c_{6},c_{7},c_{8},c_{9},c_{10},c_{11},c_{12},z_{1},z_{2},
z_{3},z_{4},z_{5},z_{6},z_{7},z_{8},z_{9},z_{10},$ $z_{11},z_{12}, z_{13},z_{14},z_{15},z_{16},z_{17},$ $z_{18},z_{19},z_{20},z_{21},z_{22},z_{23},z_{24},M_{U})$ to $\mathcal{SP}$.
\medskip

\item $\mathcal{SP}$ checks
$c_{1}\stackrel{?}{=}\mathcal{H}(E_{1}||E_{1}'||M_{U})$, $c_{2}\stackrel{?}{=}\mathcal{H}(E_{2}||E'_{2}||M_{U})$, $c_{3}\stackrel{?}{=}\mathcal{H}(\Theta_{5}||\Theta_{6}||M_{U})$, $c_{4}\stackrel{?}{=}\mathcal{H}(\Theta_{7}||\Theta_{8}||M_{U})$, $c_{5}\stackrel{?}{=}\mathcal{H}(\Theta_{9}||\Theta_{15}||M_{U})$, $c_{6}\stackrel{?}{=}\mathcal{H}(\Theta_{10}||\Theta_{16}||M_{U})$, $c_{7}\stackrel{?}{=}\mathcal{H}(\Theta_{1}||\Theta_{2}||M_{U})$, $c_{8}\stackrel{?}{=}\mathcal{H}(\Theta_{3}||\Theta_{4}||M_{U})$, $c_{9}\stackrel{?}{=}\mathcal{H}(\Theta_{6}||\Theta_{11}||M_{U})$, $c_{10}\stackrel{?}{=}\mathcal{H}(\Theta_{8}||\Theta_{12}||M_{U})$, $c_{11}\stackrel{?}{=}\mathcal{H}(\Theta_{9}||\Theta_{13}||M_{U})$, $c_{12}\stackrel{?}{=}\mathcal{H}(\Theta_{10}||\Theta_{14}||M_{U})$, \\

 $E'_{1}\stackrel{?}{=}\mathfrak{g}^{z_{1}}g_{1}^{z_{2}}E_{1}^{c_{1}}$, $E'_{2}\stackrel{?}{=}\mathfrak{g}^{z_{2}}h_{1}^{z_{4}}E_{2}^{c_{2}}$, $\Theta_{5}\Theta_{6}\stackrel{?}{=}e(g_{1},I_{1})^{z_{6}}\cdot e(g_{1},g_{1})^{z_{5}}\cdot e(W_{1}^{-1},I_{1})^{c_{3}}$, $\Theta_{7}\Theta_{8}\stackrel{?}{=}e(h_{1},I_{2})^{z_{7}}\cdot e(h_{1},h_{1})^{z_{8}}\cdot e((W'_{1})^{-1},I_{2})^{c_{4}}$, $\Theta_{9}\Theta_{15}\stackrel{?}{=}e(g_{1},g_{1})^{z_{9}}\cdot e(g_{1},I_{3})^{z_{10}}\cdot(e(I_{3},W_{1}^{-1})\cdot e(g_{1},I_{3})^{-l})^{c_{5}}$,
 $\Theta_{10}\Theta_{16}\stackrel{?}{=}e(h_{1},h_{1})^{z_{11}}\cdot e(h_{1},I_{4})^{z_{12}}\cdot(e(I_{4},(W'_{1})^{-1})\cdot e(h_{1},I_{4})^{-k})^{c_{6}}$, $\Theta_{1}\Theta_{2}\stackrel{?}{=}e(\mathfrak{g},F_{2})^{z_{13}}\cdot e(g_{2},g_{2})^{z_{14}}\cdot e(F_{1}W_{2},F_{2})^{c_{7}}$,
 $\Theta_{3}\Theta_{4}\stackrel{?}{=}e(\mathfrak{g},J_{2})^{z_{15}}\cdot e(h_{2},h_{2})^{z_{16}}\cdot e(J_{1}W'_{2},J_{2})^{c_{8}}$, $\Theta_{6}\Theta_{11}\stackrel{?}{=} e(g_{1},g_{1})^{z_{17}}\cdot e(\mathfrak{g},I_{1})^{z_{18}} \cdot (e(E_{1}W_{1},I_{1}))^{c_{9}}$,  $\Theta_{8}\Theta_{12}\stackrel{?}{=}e(h_{1},h_{1})^{z_{19}}\cdot e(\frak{g},I_{2})^{z_{20}}\cdot  e(E_{2}W'_{1},I_{2})^{c_{10}}$, 
 $\Theta_{9}\Theta_{13}\stackrel{?}{=}e(g_{1},g_{1})^{z_{21}}\cdot e(\frak{g},I_{3})^{z_{22}}\cdot (e(E_{1}g_{1}^{l}W_{1},I_{3}))^{c_{11}}$, 
 $\Theta_{10}\Theta_{14}\stackrel{?}{=}e(h_{1},h_{1})^{z_{23}}\cdot e(\frak{g},I_{4})^{z_{24}}\cdot  (e(E_{2}h_{1}^{k}W'_{1},I_{4}))^{c_{12}}$,


%
 
\end{enumerate}
\medskip

\noindent{\em An Instance of Zero Knowledge Proof $\prod_{SP}^{2}$.}

\begin{enumerate}
\item For $\mu\in\{1,2,\cdots,l\}$ and $\nu\in\{1,2,\cdots,k\}$, $\mathcal{SP}$ selects $\omega_{\mu},\psi_{\nu}\stackrel{R}{\leftarrow}\mathbb{Z}_{p}$, $\tilde{\mathfrak{h}}\stackrel{R}{\leftarrow}\mathbb{G}$ and $M_{SP}^{2}\stackrel{R}{\leftarrow}\{0,1\}^{*}$, and computes $K_{\mu,\nu}=E_{1}g_{1}^{\mu}E_{2}h_{1}^{\nu}F_{1}^{x^{\mu}}J_{1}^{y^{\nu}}$, $L_{\mu,\nu}=e(K_{\mu,\nu},\mathfrak{h})$, $\Upsilon_{\mu,\nu}^{1}=F_{1}^{\omega_{\mu}}J_{1}^{\psi_{\nu}}$, $\Upsilon_{\mu,\nu}^{2}=e(C_{\mu,2},g_{2})^{-\omega_{\mu}}$, $\Upsilon_{\mu,\nu}^{3}=e(D_{\nu,2},h_{2})^{-\psi_{\nu}}$, \\
 $c_{\mu,\nu}^{1}=\mathcal{H}(K_{\mu,\nu}||\Upsilon_{\mu,\nu}^{1}||\Upsilon_{\mu,\nu}^{2}||\Upsilon_{\mu,\nu}^{3}||M_{SP}^{2})$, $c_{\mu,\nu}^{2}=\mathcal{H}(L_{\mu,\nu}||H||M_{SP}^{2})$, 
$z_{\mu,\nu}^{1}=\omega_{\mu}-c_{\mu,\nu}^{1}x^{\mu}$,   $z_{\mu,\nu}^{2}=\psi_{\nu}-c_{\mu,\nu}^{1}y^{\nu}$, $\mathfrak{h}_{\mu,\nu}=\tilde{\mathfrak{h}}\mathfrak{h}^{-c_{\mu,\nu}^{2}}$, $\tilde{H}=e(\mathfrak{g},\tilde{\mathfrak{h}})$, $L'_{\mu,\nu}=e(K_{\mu,\nu},\tilde{\mathfrak{h}})$.

$\mathcal{SP}$ sends $(K_{\mu,\nu},L_{\mu,\nu},\Upsilon_{\mu,\nu}^{1},\Upsilon_{\mu,\nu}^{2},\Upsilon_{\mu,\nu}^{3},c_{\mu,\nu}^{1},c_{\mu,\nu}^{2},z_{\mu,\nu}^{1},z_{\mu,\nu}^{2},\mathfrak{h}_{\mu,\nu},\tilde{H},L'_{\mu,\nu})$ to $\mathcal{U}$.
\medskip

\item $\mathcal{U}$ checks  $c_{\mu,\nu}^{1}\stackrel{?}{=}\mathcal{H}(K_{\mu,\nu}||\Upsilon_{\mu,\nu}^{1}||\Upsilon_{\mu,\nu}^{2}||\Upsilon_{\mu,\nu}^{3}||M_{SP}^{2})$, $c_{\mu,\nu}^{2}\stackrel{?}{=}\mathcal{H}(L_{\mu,\nu}||H||M_{SP}^{2})$, $\Upsilon_{\mu,\nu}^{1}\stackrel{?}{=}F_{1}^{z_{\mu,\nu}^{1}}J_{1}^{z_{\mu,\nu}^{2}}(\frac{K_{\mu,\nu}}{E_{1}g_{1}^{\mu}E_{2}h_{1}^{\nu}})^{c_{\mu,\nu}^{1}}$, $\Upsilon_{\mu,\nu}^{2}\stackrel{?}{=}e(C_{\mu,2},g_{2})^{-z_{\mu,\nu}^{1}}(\frac{e(C_{\mu,2},W_{2})}{e(g_{2},g_{2})})^{c_{\mu,\nu}^{1}}$,  $\Upsilon_{\mu,\nu}^{3}\stackrel{?}{=}e(D_{\nu,2},h_{2})^{-z_{\mu,\nu}^{2}}(\frac{e(D_{\nu,2},W'_{2})}{e(h_{2},h_{2})})^{c_{\mu,\nu}^{2}}$, $L'_{\mu,\nu}\stackrel{?}{=} e(K_{\mu,\nu},\mathfrak{h}_{\mu,\nu})\cdot L_{\mu,\nu}^{c_{\mu,\nu}^{2}}$,  $\tilde{H}\stackrel{?}{=}e(\mathfrak{g},\mathfrak{h}_{\mu,\nu})\cdot H^{c_{\mu,\nu}^{2}}$.
\end{enumerate}

\end{document}